\documentclass[a4paper,11pt]{article}
\usepackage{latexsym,amssymb,amsmath,amsthm}
\usepackage{fullpage}
\usepackage{graphicx}
\theoremstyle{plain}
\newtheorem{theorem}{Theorem}[section]
\newtheorem{proposition}{Proposition}[section]
\newtheorem{corollary}{Corollary}[section]
\newtheorem{lemma}{Lemma}[section]
\theoremstyle{definition}
\newtheorem{definition}{Definition}[section]
\theoremstyle{remark}
\newtheorem{example}{Example}[section]
\newtheorem{remark}{Remark}[section]

\title{Optimal Observables for Minimum-Error State Discrimination in General Probabilistic Theories}
\author{Koji Nuida, Gen Kimura, Takayuki Miyadera}
\date{Research Center for Information Security (RCIS), National Institute of Advanced Industrial Science and Technology (AIST)\\%
Akihabara-Daibiru Room 1003, 1-18-13 Sotokanda, Chiyoda-ku, Tokyo 101-0021, Japan\\%
E-mail: \{k.nuida, gen-kimura, miyadera-takayuki\}[at]aist.go.jp%
}
\begin{document}
\maketitle
\begin{abstract}
General Probabilistic Theories provide the most general mathematical framework for the theory of probability in an operationally natural manner, and generalize classical and quantum theories.
In this article, we study state-discrimination problems in general probabilistic theories using a Bayesian strategy.
After re-formulation of the theories with mathematical rigor, we first prove that an optimal observable to discriminate any (finite) number of states always exists in the most general setting.
Next, we revisit our recently proposed geometric approach for the problem and show that, for two-state discrimination, this approach is indeed effective in arbitrary dimensional cases.
Moreover, our method reveals an operational meaning of Gudder's \lq\lq intrinsic metric'' by means of the optimal success probability, which turns out to be a generalization of the trace distance for quantum systems.
As its by-product, an information-disturbance theorem in general probabilistic theories is derived, generalizing its well known quantum version.
\end{abstract}
\section{Introduction}
\label{sec:intro}

\subsection{Background}
\label{subsec:intro_background}
Among many attempts to understand quantum theory axiomatically, an operationally natural approach for the general theory of probability, recently referred to as \emph{general probabilistic theories} (or \emph{generic probabilistic models}), has been studied \cite{Gud_book88,Gud_book79,Hol_book,Oza80} and has attracted much attention in the recent development of quantum information theory (e.g., \cite{BBLW07,ref:Dariano,KMI08}). 
Such an approach provides a unified mathematical framework that involves not only classical and quantum theories but also more general settings that would be candidates of possible future extensions of the present quantum theory.
One of the motivations of such an approach is to understand quantum mechanics better by introducing various viewpoints especially with information theoretic point of view. 
Another motivation to investigate such a general theory has arisen recently from research on quantum information theory including quantum information security.
Among recent development of information theory and information security, one of the greatest impacts was provided by Shor's discovery \cite{Sho94} of an efficient (i.e., polynomial-time) integer factoring algorithm for quantum computers that reveals a future practical threat against several standard cryptosystems in the present time, such as RSA cryptosystem \cite{RSA78}.
This history suggests a non-negligible possibility that any cryptosystem with security based on the present physical theory, even quantum theory, may fall into insecure once a further advanced physical theory is discovered and applied to information technology.
Hence a study of possible extensions of the present physical theory is of importance and interest from not only theoretical but also practical standpoints.

One of the most important aims of studying general probabilistic theories is to determine which characteristics are typical for classical or quantum systems and which are not.
For example, in a recent article \cite{BBLW07} Barnum et al.\ investigated cloning and broadcasting of states in a general probabilistic theory.
They proved (in finite-dimensional cases) that universal cloning or universal broadcasting is possible only for classical systems, which generalizes the No-Cloning Theorem and the No-Broadcasting Theorem for quantum systems \cite{BCFJS96,Die82,WZ82,Yue86}.
Another example relevant to our present work is our recent study \cite{KMI08} on minimum-error state discrimination problems in general probabilistic theories (in this article the word \lq\lq minimum-error'' is omitted since we do not discuss other kinds of discrimination problems such as unambiguous state discrimination).
State discrimination problems have been well investigated for quantum systems (e.g., \cite{BKMH97,Hel_book,JRF02,YKL75}), but optimal success probabilities to discriminate given states and the corresponding optimal measurements were determined only in very restricted cases such as two-state cases.
In \cite{KMI08} we gave a formulation of state discrimination problems in finite-dimensional general probabilistic theories, and introduced from a geometric viewpoint a class of special ensembles of states called \emph{Helstrom families}: We showed that the optimal success probability can be determined by a Helstrom family if it exists. 
For the existence, we have discussed only for two-state cases and some other cases of states with symmetric configuration, and it has been shown that a Helstrom family always exists for both classical and quantum systems in any \lq\lq generic'' case (specified in a certain well-defined manner).
However, existence of Helstrom families in more general (neither classical nor quantum) cases has not been clarified.
The main aim of this article is to study the existence problem of the Helstrom family in general probabilistic theories with arbitrary dimension that are neither classical nor quantum. 

\subsection{Our contributions and organization of the article}
\label{subsec:intro_contribution}

In Sect.\ \ref{sec:formulation_model}, we summarize a mathematical framework for general probabilistic theories.
Following several preceding works for general probabilistic theories (e.g., \cite{BBLW07,Gud_book79,Hol_book,KMI08,Mac_book,Oza80}), our formulation is based on the notions of \emph{states}, \emph{effects} and \emph{observables}, as well as the notion of probabilistic state ensembles.
Namely, we regard the state space as a \lq\lq convex structure'' \cite{Gud_book79}.
A standard argument shows that the associated \lq\lq separated'' state space is embedded as a convex subset $\mathcal{S}$ in a real vector space $V$.
For the sake of minimality, we assume that $V$ is the affine hull of $\mathcal{S}$ and the topology of $V$ is the weak topology generated by all effects on $\mathcal{S}$.
We emphasize that $\mathcal{S}$ is usually assumed to be compact with respect to this topology, but in the present article compactness is not assumed to keep the most generality of our setting. 
In fact, when $\mathcal{S}$ is not compact with respect to this topology, we further take a \lq\lq virtual state space'' $\widetilde{\mathcal{S}} \supset \mathcal{S}$ and a \lq\lq virtual underlying space'' $\widetilde{V} \supset V$ such that $\widetilde{\mathcal{S}}$ is a \emph{compact} convex subset of $\widetilde{V}$ and some additional conditions are satisfied (see Theorem \ref{thm:state_space_embedded} for the precise statement):
\begin{displaymath}
\begin{matrix}
& \mathrm{cl}_{\widetilde{V}}(\mathcal{S}) = & \widetilde{\mathcal{S}} & \quad \subset \quad & \widetilde{V} \\
& & \cup & & \cup \\
\overline{\mathcal{S}_0} = \mathcal{S}_0/{\sim} & \quad \simeq \quad & \mathcal{S} & \quad \subset \quad & V 
\end{matrix}
\end{displaymath}
By those properties, the objects $V$, $\widetilde{\mathcal{S}}$ and $\widetilde{V}$ are uniquely determined by $\mathcal{S}$, called the \emph{minimal framework}.
See Appendices \ref{sec:appendix_S_V}--\ref{sec:appendix_uniqueness} for further technical details.
Now the effects on the \lq\lq real'' state space $\mathcal{S}$ are in one-to-one correspondence to their continuous extensions to $\widetilde{\mathcal{S}}$, called \lq\lq virtual effects''.
A similar correspondence exists between observables on $\mathcal{S}$ and \lq\lq virtual observables'' on $\widetilde{\mathcal{S}}$.
Moreover, for each \lq\lq virtual state'' $\widetilde{s} \in \widetilde{\mathcal{S}} \setminus \mathcal{S}$, any $\varepsilon > 0$ and any observables $\mathbf{O}_1,\dots,\mathbf{O}_k$, there exists a \lq\lq real state'' $s \in \mathcal{S}$ such that the results of measurements of these $\mathbf{O}_i$ at $\widetilde{s}$ are within $\varepsilon$-error from the results at $s$; physically, this means that virtual states and real states are indistinguishable by experiments.
Note that, in finite-dimensional cases, the underlying space $V$ is always isomorphic to a finite-dimensional Euclidean space and now $\mathcal{S}$ is nothing but a bounded convex subset of the Euclidean space $V$.

In Sect.\ \ref{sec:discrimination_problem}, we give a natural formulation of state discrimination problems in general probabilistic theories by following our preceding work \cite{KMI08}.
Our present formulation coincides with the preceding one when $\mathcal{S}$ is compact.
Moreover, we show that an optimal observable always exists for discrimination of any (finite) number of given states with arbitrary a priori occurrence probabilities (see Theorem \ref{thm:optimal_observable_exist}).
Although it would be possible to interpret this result as a special case of a general theorem by Ozawa \cite{Oza80}, we include the proof in this article for the reader's convenience because of its simplicity.  
(The proof uses only the existence theorem of maximum values of continuous functions on compact spaces and some elementary arguments for topological spaces.) 
Note that the argument in Sect.\ \ref{sec:discrimination_problem} is closed within the real state space $\mathcal{S}$, therefore the additional notions such as virtual states and virtual observables are not yet needed.

In Sect.\ \ref{sec:Helstrom}, we introduce the notion of (weak) Helstrom families by translating the definition given in \cite{KMI08} to our minimal framework.
A weak Helstrom family yields an upper bound of the optimal success probability for discriminating given states, while a Helstrom family yields the tight bound.
A sufficient condition for a weak Helstrom family to be a Helstrom family has been given \cite{KMI08}.
As a consequence of the above-mentioned existence theorem of an optimal observable, we show that the above sufficient condition is also necessary, except for the meaningless cases called \emph{non-generic cases}.
(By definition, \emph{generic cases} are the cases where there exists a discrimination strategy better than simply outputting the candidate state with highest a priori probability.)
In two-state cases, the above necessary and sufficient condition turns out to be \lq\lq distinguishability'' of two (possibly virtual) states $t_1,t_2$ associated to a given weak Helstrom family, therefore the problem of finding a Helstrom family is reduced to a study of distinguishable (virtual) states. 

Finally, in Sect.\ \ref{sec:existence} we prove that a Helstrom family for two-state discrimination always exists in generic cases (see Theorem \ref{thm:HE_exist}), hence in such a case the optimal success probability can be determined (at least in principle) by just finding a Helstrom family.
Our argument works in a general case of arbitrary dimension that may be neither classical nor quantum.
Owing to the result, we also give a simple criterion for generic cases among all two-state cases (see Theorem \ref{thm:criterion_generic}): Given two distinct candidate states $s_1,s_2 \in \mathcal{S}$ with positive a priori probabilities $p_1,p_2$, the case is non-generic if and only if we have $p_1 \neq p_2$ and an element $s^{\ast} = (p_1 s_1 - p_2 s_2)/(p_1 - p_2)$ of $V$ lies outside the state space $\mathcal{S}$.
In particular, the equiprobable cases $p_1 = p_2 = 1/2$ are always generic, therefore in such cases we are always able to discriminate (at least in principle) given states with probability higher than $1/2$.
Moreover, our result also reveals a relation of Gudder's distance between two states $s_1,s_2 \in \mathcal{S}$ \cite{Gud_book79} with the optimal success probability of discriminating $s_1$ and $s_2$ in equiprobable cases, and also an operational meaning of Gudder's intrinsic metric \cite{Gud_book79} that gives an operationally natural generalization of the trace distance for quantum systems to general probabilistic theories (see Remark \ref{rem:Gudder_distance_optimal_prob}). 
As an application, a simple (qualitative) version of the information disturbance theorem in general probabilistic theories is shown to be hold that generalizes the corresponding theorem in quantum theory.   

\section{A Mathematical Framework for General Probabilistic Theories}
\label{sec:formulation_model}

In this section, we introduce a mathematical framework for general probabilistic theories.
In this article, any vector space is defined over the real field $\mathbb{R}$ unless otherwise specified.

Following the preceding works \cite{BBLW07,Gud_book79,Hol_book,KMI08,Mac_book,Oza80}, we start with a set $\mathcal{S}_0$ of \emph{states}, called a \emph{state space}, that is a \emph{convex structure} \cite{Gud_book79} in the following sense: For two states $s,t \in \mathcal{S}_0$ and two weights $\lambda,\mu \geq 0$ such that $\lambda + \mu = 1$, a state $\langle \lambda,\mu;s,t \rangle \in \mathcal{S}_0$ called an \emph{ensemble} of $s,t$ with weights $\lambda,\mu$ is uniquely determined.
Physically, $\langle \lambda,\mu;s,t \rangle$ means the probabilistic state ensemble of $s$ and $t$ with a priori probabilities $\lambda$ and $\mu$.
We regard any convex subset of a vector space as a convex structure with a natural operation $\langle \lambda,\mu;s,t \rangle = \lambda s + \mu t$.
Note that any other postulate for the operation $\langle \lambda,\mu;s,t \rangle$ is not required; some natural properties of state ensembles will be induced by construction of the associated \lq\lq separated'' state space presented below.

For any convex structure $C$, we say that a functional $f:C \to \mathbb{R}$ on $C$ is \emph{affine} if we have $f(\langle \lambda,\mu;s,t \rangle) = \lambda f(s) + \mu f(t)$ for any $s,t \in C$.
Let $\mathcal{E}(C)$ denote the set of all affine functionals $e$ on $C$ with image $e(C)$ contained in the unit interval $\left[0,1\right]$ in $\mathbb{R}$.
Then we call each $e \in \mathcal{E}(\mathcal{S}_0)$ an \emph{effect} on $\mathcal{S}_0$.
Now we define an equivalence relation $\sim$ on $\mathcal{S}_0$ by setting $s \sim t$ if and only if $e(s) = e(t)$ for every $e \in \mathcal{E}(\mathcal{S}_0)$.
Let $\overline{s}$ denote the equivalence class of $s \in \mathcal{S}_0$.
Then the quotient set $\overline{\mathcal{S}_0} = \mathcal{S}_0 /{\sim}$ is also a convex structure with $\langle \lambda,\mu;\overline{s},\overline{t} \rangle = \overline{\langle \lambda,\mu;s,t \rangle}$ for any $\overline{s},\overline{t} \in \overline{\mathcal{S}_0}$.
A physical interpretation is that, as two equivalent states (in the above sense) are statistically indistinguishable for any effect, we would have no physical way to distinguish those states.
(See below for the definition of observables composed of effects.)
Now each $e \in \mathcal{E}(\mathcal{S}_0)$ induces an effect $\overline{e} \in \mathcal{E}(\overline{\mathcal{S}_0})$ on $\overline{\mathcal{S}_0}$ by $\overline{e}(\overline{s}) = e(s)$ for each $\overline{s} \in \overline{\mathcal{S}_0}$, and this defines a one-to-one correspondence between $\mathcal{E}(\mathcal{S}_0)$ and $\mathcal{E}(\overline{\mathcal{S}_0})$.
Moreover, the definition of the set $\overline{\mathcal{S}_0}$ implies the following property (see e.g., \cite{Gud_book79,Hol_book,Oza80}):
\begin{lemma}
\label{lem:separated}
The convex structure $\overline{\mathcal{S}_0}$ is separated, in the sense that for any distinct $s,t \in \overline{\mathcal{S}_0}$, there exists an effect $e \in \mathcal{E}(\overline{\mathcal{S}_0})$ such that $e(s) \neq e(t)$.
\end{lemma}

The next theorem presents our framework involving the separated state space $\overline{\mathcal{S}_0}$.
To our framework we intend to introduce as few mathematical structures as possible subject to physically natural requirements; we call the resulting framework a \emph{minimal framework}.
Here we use the notion of topological vector spaces; we refer to the book \cite{SW_book} for theory of topological vector spaces together with some relevant topics in general topology.
In what follows, we abbreviate \lq\lq topological vector space'' to \lq\lq t.v.s.'', and \lq\lq locally convex'' to \lq\lq l.c.''.
For any t.v.s.\ $W$, let $\mathcal{L}_c(W)$ denote the set of all continuous linear functionals $W \to \mathbb{R}$.
Moreover, let $\mathcal{T}(X)$ denote the topology on a set $X$ if it is clear from the context.
Then the above-mentioned theorem on our minimal framework is the following:
\begin{theorem}
\label{thm:state_space_embedded}
Given a separated convex structure $\overline{\mathcal{S}_0}$ as above, there exist the following objects:
\begin{itemize}
\item a l.c.\ Hausdorff t.v.s.\ $\widetilde{V}$ (over $\mathbb{R}$);
\item a convex subset $\widetilde{\mathcal{S}}$ of $\widetilde{V}$ such that $\widetilde{V}$ is the affine hull $\mathrm{Aff}(\widetilde{\mathcal{S}})$ of $\widetilde{\mathcal{S}}$;
\item a topological vector subspace $V$ of $\widetilde{V}$ that is dense in $\widetilde{V}$;
\item a convex subset $\mathcal{S}$ of $V$ such that $\mathrm{Aff}(\mathcal{S}) = V$,
\end{itemize}
satisfying the following conditions:
\begin{itemize}
\item $\mathcal{S}$ is isomorphic to $\overline{\mathcal{S}_0}$, in the sense that there exists a bijection $\varphi:\overline{\mathcal{S}_0} \to \mathcal{S}$ such that $\varphi(\langle \lambda,\mu;s,t \rangle) = \lambda \varphi(s) + \mu \varphi(t)$ for any $s,t \in \overline{\mathcal{S}_0}$;
\item the topology $\mathcal{T}(\widetilde{V})$ of $\widetilde{V}$ is a weak topology, i.e.,  the topology with minimal family of open subsets to make every $f \in \mathcal{L}_c(\widetilde{V})$ continuous;
\item the induced topology on $\mathcal{S}$ is the weakest to make every $e \in \mathcal{E}(\mathcal{S})$ continuous;
\item the induced topology on $V$ is the weakest to make every linear functional $f:V \to \mathbb{R}$, such that $f(\mathcal{S}) \subset \mathbb{R}$ is bounded, a continuous map;
\item $\widetilde{\mathcal{S}}$ is the closure $\mathrm{cl}_{\widetilde{V}}(\mathcal{S})$ of $\mathcal{S}$ in $\widetilde{V}$, and $\widetilde{\mathcal{S}}$ is compact and complete.
\end{itemize}
\begin{displaymath}
\begin{matrix}
& \mathrm{cl}_{\widetilde{V}}(\mathcal{S}) = & \widetilde{\mathcal{S}} & \quad \subset \quad & \widetilde{V} \\
& & \cup & & \cup \\
\overline{\mathcal{S}_0} = \mathcal{S}_0/{\sim} & \quad \simeq \quad & \mathcal{S} & \quad \subset \quad & V 
\end{matrix}
\end{displaymath}
Moreover, these objects are unique; namely, for another collection $\mathcal{S}'$, $V'$, $\widetilde{\mathcal{S}}'$ and $\widetilde{V}'$ of such objects, there exists an affine isomorphism $\widetilde{V} \to \widetilde{V}'$ that is a homeomorphism and maps each of $\mathcal{S}$, $V$ and $\widetilde{\mathcal{S}}$ onto the corresponding object.
\end{theorem}

A proof of Theorem \ref{thm:state_space_embedded} will be given in Appendices \ref{sec:appendix_S_V}--\ref{sec:appendix_uniqueness}.
\begin{remark}
\label{rem:finite_dim}
In finite-dimensional cases ($\dim\mathcal{S} = n < \infty$), the space $V$ above is isomorphic to an $n$-dimensional Euclidean space $\mathbb{R}^n$ (cf., Theorem \ref{thm:fin_dim_is_Euclidean}), and we have $\widetilde{V} = V$ and $\widetilde{\mathcal{S}} = \mathrm{cl}_V(\mathcal{S})$.
Hence in such cases, the state space $\mathcal{S}$ is nothing but a bounded convex subset of $\mathbb{R}^n$.
Moreover, in this case every $e \in \mathcal{E}(\widetilde{\mathcal{S}})$ is continuous by the definition of $\mathcal{T}(V) = \mathcal{T}(\widetilde{V})$; however, the continuity is not guaranteed in a general case.
\end{remark}
\begin{definition}
\label{defn:state_terminology}
We call the sets $\mathcal{S}$, $\widetilde{\mathcal{S}}$, $V$, and $\widetilde{V}$ a \emph{(real) state space}, a \emph{virtual state space}, a \emph{(real) underlying space}, and a \emph{virtual underlying space}, respectively.
We call $s \in \mathcal{S}$ a \emph{(real) state} and $\widetilde{s} \in \widetilde{\mathcal{S}} \setminus \mathcal{S}$ a \emph{virtual state}.
Moreover, we call each $e \in \mathcal{E} = \mathcal{E}(\mathcal{S})$ a \emph{(real) effect} on $\mathcal{S}$, and each $\widetilde{e} \in \mathcal{E}(\widetilde{\mathcal{S}})$ a \emph{virtual effect} on $\widetilde{\mathcal{S}}$ if it is continuous.
Let $\widetilde{\mathcal{E}}$ denote the set of the virtual effects on $\widetilde{\mathcal{S}}$, i.e., $\widetilde{\mathcal{E}} = \{\widetilde{e} \in \mathcal{E}(\widetilde{\mathcal{S}}) \mid \widetilde{e} \mbox{ is continuous}\}$.
\end{definition}
The choice of $\mathcal{T}(\mathcal{S})$ is motivated by a physical intuition that any available information on the state space $\mathcal{S}$ would be obtained via statistical properties of effects on $\mathcal{S}$.
On the other hand, the continuity of virtual effects are required to ensure the following correspondence between effects and virtual effects:
\begin{lemma}
\label{lem:effect_correspondence}
Each effect $e \in \mathcal{E}$ on $\mathcal{S}$ has a unique continuous affine extension $\widetilde{e}:\widetilde{\mathcal{S}} \to \mathbb{R}$, and we have $\widetilde{e} \in \widetilde{\mathcal{E}}$.
This gives a bijection $e \mapsto \widetilde{e}$ from $\mathcal{E}$ to $\widetilde{\mathcal{E}}$.
\end{lemma}
\begin{proof}
Only the nontrivial part is the existence of a continuous affine extension $\widetilde{e}$ of $e$ with $\widetilde{e} \in \widetilde{\mathcal{E}}$; the uniqueness then follows since $\mathcal{S}$ is dense in $\widetilde{\mathcal{S}}$.
First, since $\mathrm{Aff}(\mathcal{S}) = V$, the effect $e$ extends to an affine functional $f:V \to \mathbb{R}$.
Let $\alpha$ be the value of $f$ at the origin of $V$; therefore $f' = f - \alpha:V \to \mathbb{R}$ is linear.
Note that $f'(\mathcal{S}) \subset \left[-\alpha,1-\alpha\right]$, therefore $f'$ is continuous on $V$ by the property of $V$ in Theorem \ref{thm:state_space_embedded}.
By a consequence of Hahn-Banach's Theorem (Theorem \ref{thm:extension_continuous_to_continuous}), this $f'$ extends to a continuous linear functional $g$ on $\widetilde{V}$.
Now $g(\widetilde{\mathcal{S}}) \subset \mathrm{cl}_{\mathbb{R}}(g(\mathcal{S}))$ since $\widetilde{\mathcal{S}} = \mathrm{cl}_{\widetilde{V}}(\mathcal{S})$, while $g(\mathcal{S}) \subset \left[-\alpha,1-\alpha\right]$ since $g$ is an extension of $f'$.
Thus the restriction $\widetilde{e} = (g + \alpha)|_{\widetilde{\mathcal{S}}}$ of $g + \alpha$ to $\widetilde{\mathcal{S}}$ is a continuous affine functional such that $\widetilde{e}(\widetilde{\mathcal{S}}) \subset \left[0,1\right]$, therefore $\widetilde{e} \in \widetilde{\mathcal{E}}$.
This $\widetilde{e}$ is the desired extension of $e$.
\end{proof}
Moreover, the sets $\mathcal{S}$ and $\widetilde{\mathcal{S}}$ have the following properties:
\begin{lemma}
\label{lem:separated_generalized}
Both $\mathcal{S}$ and $\widetilde{\mathcal{S}}$ are separated, which (for $\widetilde{\mathcal{S}}$) means that for any distinct $s,t \in \widetilde{\mathcal{S}}$, there exists an $e \in \widetilde{\mathcal{E}}$, not just $e \in \mathcal{E}(\widetilde{\mathcal{S}})$, such that $e(s) \neq e(t)$.
\end{lemma}
\begin{proof}
Since $V$ is Hausdorff, $\mathcal{S}$ is separated (in the sense of Lemma \ref{lem:separated}) by the definition of $\mathcal{T}(\mathcal{S})$; see Theorem \ref{thm:state_space_embedded}.
On the other hand, let $s,t$ be distinct elements of $\widetilde{\mathcal{S}}$.
Then, since $\mathcal{T}(\widetilde{V})$ is Hausdorff and a weak topology (see Theorem \ref{thm:state_space_embedded}), there exists a continuous linear functional $f$ on $\widetilde{V}$ such that $f(s) \neq f(t)$.
Now $f(\widetilde{\mathcal{S}}) \subset \mathbb{R}$ is bounded since $\widetilde{\mathcal{S}}$ is compact, therefore the restriction $e$ of an appropriate affine transformation $\alpha f + \beta$ of $f$ to $\widetilde{\mathcal{S}}$, where $\alpha,\beta \in \mathbb{R}$, is a virtual effect such that $e(s) \neq e(t)$.
Hence Lemma \ref{lem:separated_generalized} holds.
\end{proof}
\begin{definition}
\label{defn:observable}
An $N$-valued \emph{(real) observable} (or \emph{virtual observable}, respectively) is a collection $\mathbf{O} = (e_i)_{i = 1}^{N}$ of $N$ effects $e_i \in \mathcal{E}$ (or $N$ virtual effects $e_i \in \widetilde{\mathcal{E}}$, respectively) such that $\sum_{i = 1}^{N} e_i = 1$.
Let $\mathcal{O}_N$ and $\widetilde{\mathcal{O}}_N$ denote the sets of all $N$-valued observables and of all $N$-valued virtual observables, respectively. 
\end{definition}
Physically, for each observable $\mathbf{O} = (e_i)_i$ and each $s \in \mathcal{S}$, the quantity $e_i(s)$ means the probability to obtain $i$-th output when measuring $\mathbf{O}$ at the state $s$; the condition $\sum_i e_i = 1$ is required by a property of probability.
On the other hand, the affine property of each $e_i$ is motivated by a natural expectation that the output probabilities for a probabilistic state ensemble would be weighted sums of those probabilities for each of the original state.
The same also holds for virtual observables.
Now we have the following correspondence:
\begin{lemma}
\label{lem:observable_correspondence}
We have $\widetilde{\mathbf{O}} = (\widetilde{e_i})_i \in \widetilde{\mathcal{O}}_N$ for any $\mathbf{O} = (e_i)_i \in \mathcal{O}_N$.
This gives a bijection $\mathbf{O} \mapsto \widetilde{\mathbf{O}}$ from $\mathcal{O}_N$ to $\widetilde{\mathcal{O}}_N$.
\end{lemma}
\begin{proof}
Only the nontrivial part is to show that $\sum_i \widetilde{e_i} = 1$ for any $\mathbf{O} = (e_i)_i \in \mathcal{O}_N$.
This follows from the uniqueness property in Lemma \ref{lem:effect_correspondence}, since both $\sum_i \widetilde{e_i}$ and $1$ are continuous affine extensions of the effect $\sum_i e_i = 1$ to $\widetilde{\mathcal{S}}$.
\end{proof}
By virtue of Lemma \ref{lem:observable_correspondence}, the output probabilities for virtual observables at virtual states can be derived (at least in principle) from information on real observables at real states.
On the other hand, for any finite collection of measurements with non-ideal accuracy, virtual states are indistinguishable from real states (in the sense mentioned in Sect.\ \ref{subsec:intro_contribution}).

Note that our framework presented above does in fact not concern every feature of quantum theory, e.g., transformations of states possibly caused by measuring observables.
However, our framework is still enough for our current purpose of studying state discrimination problems.

Obviously, two fundamental examples of general probabilistic theories are given by classical and quantum theories, as follows (taken from \cite{KMI08}):
\begin{example}
\label{exmp:classical}
A finite classical system described by a finite probability theory with finite sample space $\{\omega_1,\dots,\omega_n\}$ is formulated in our model as the ($n-1$)-dimensional standard simplex $\mathcal{S} = \{p = (p_1,\dots,p_n) \in \mathbb{R}^n \mid p_i \geq 0, \sum_i p_i = 1\}$.
Namely, each state is a probability distribution over the sample space, and it can be seen as a probabilistic ensemble of \lq\lq pure states'' $p^{(i)}$, $i = 1,\dots,n$, with only one possible output $\omega_i$, that are extremal points of $\mathcal{S}$ in usual sense.
Note that in this example $\mathcal{S}$ itself is compact, hence all states are real.
This example can be naturally extended to infinite-dimensional classical systems.
\end{example}
\begin{example}
\label{exmp:quantum}
In quantum theory, a quantum state is described by a density operator $\rho$, that is a positive operator on a Hilbert space $\mathcal{H}$ with unit trace.
Thus the state space is a convex subset of the vector space of all linear operators on $\mathcal{H}$.
Moreover, an effect $e$ is described \cite{Oza80} by a positive bounded operator $B$ such that $0 \leq B \leq I_{\mathcal{H}}$ via the relation $e(\rho) = \mathrm{tr} B \rho$, that is an element of positive operator valued measure (POVM).
\end{example}

In the last of this section, we give two remarks on relations with preceding works.
Before starting the remarks, note that assumptions on compactness of the state space $\mathcal{S}$ and on completeness of $\mathcal{S}$ are equivalent to each other, since each of the two implies that $\mathcal{S}$ is closed in $\widetilde{V}$ and hence $\widetilde{\mathcal{S}} = \mathcal{S}$.
\begin{remark}
\label{rem:pure_states}
In a recent work by Barnum et al.\ \cite{BBLW07}, their finite-dimensional state space is assumed to be compact to guarantee that the state space is the closed convex hull of the set of \lq\lq pure states'' (i.e., extremal points of the state space).
Owing to Krein-Milman's Theorem (see e.g., Theorem 10.4 in \cite[Chapter II]{SW_book}), the same property is possessed by our (possibly infinite-dimensional) virtual state space $\widetilde{\mathcal{S}}$.
Thus it is very attractive to start our argument by choosing the compact set $\widetilde{\mathcal{S}}$ as a new \lq\lq state space'' instead of $\mathcal{S}$.
However, such a modification \emph{does} decrease the generality of our framework.
Namely, it is \emph{not} guaranteed in general that every $e \in \mathcal{E}(\widetilde{\mathcal{S}})$, that should be a new \lq\lq effect'' in the above modification, is continuous with respect to the original topology of $\widetilde{\mathcal{S}}$.
Thus to ensure that every \lq\lq effect'' is continuous, we need a new topology stronger than the original, therefore the set $\widetilde{\mathcal{S}}$ may fail compactness with respect to the new topology.
Hence the advantage to choose $\widetilde{\mathcal{S}}$ as a state space disappears.
\end{remark}
\begin{remark}
\label{rem:Gudder_metric}
In another previous work by Gudder \cite{Gud_book79}, the following distance $d(s,s')$ of two states $s,s' \in \mathcal{S}$ was introduced to make $\mathcal{S}$ a metric space.
Namely, Gudder defined $d(s,s')$ to be the infimum of the values $0 < \lambda \leq 1$ such that $\lambda t + (1 - \lambda) s = \lambda t' + (1 - \lambda) s'$ for some states $t,t' \in \mathcal{S}$.
Since this relation implies that $(1 - \lambda) |e(s) - e(s')| = \lambda |e(t) - e(t')| \leq \lambda$ for any $e \in \mathcal{E}$, every effect is continuous with respect to the metric $d$ on $\mathcal{S}$.
However, unless $\mathcal{S}$ is finite-dimensional, the metric $d$ is \emph{not} necessarily continuous with respect to the topology of $\mathcal{S}$ specified in Theorem \ref{thm:state_space_embedded}.
This is roughly because, for a state $s \in \mathcal{S}$ and any collection of a finite number of effects $e_i$, the metric $d$ is not necessarily bounded by a sufficiently small value on the intersection of hyperplanes containing $s$ defined by the affine functionals $e_i$.
Thus our topology on $\mathcal{S}$ is weaker than (or equal to) the topology defined by the metric $d$.
Moreover, another relation of our results with Gudder's metric functions will be mentioned later (Remark \ref{rem:Gudder_distance_optimal_prob}).
\end{remark}

\section{State Discrimination Problems}
\label{sec:discrimination_problem}
In this section, we give a formulation of (minimum-error) state discrimination problems in general probabilistic theories based on the minimal framework introduced in Sect.\ \ref{sec:formulation_model}.
This formulation is a natural generalization of state discrimination problems for quantum systems, and in fact a naive translation of our preceding formulation \cite{KMI08} to the present more general setting.
 
In the state discrimination problem, we are given a finite number (say $N$) of real states $s_1,\dots,s_N \in \mathcal{S}$ and the corresponding a priori probabilities $p_1,\dots,p_N$, $p_i \geq 0$, $\sum_i p_i = 1$.
To avoid inessential intricacy, we assume that each probability $p_i$ is positive.
Then for each $N$-valued observable $\mathbf{O} = (e_i)_i \in \mathcal{O}_N$, we define the \emph{success probability} $P_{\mathrm{succ}}(\mathbf{O})$ for the observable $\mathbf{O}$ by
\begin{equation}
\label{eq:success_prob}
P_{\mathrm{succ}}(\mathbf{O}) = \sum_{i = 1}^{N} p_i e_i(s_i) \enspace.
\end{equation}
Namely, when measuring the observable $\mathbf{O}$ at an unknown state that is chosen from $s_1,\dots,s_N$ with probabilities $p_1,\dots,p_N$ (thus the unknown state is regarded as the probabilistic ensemble $\sum_i p_i s_i \in \mathcal{S}$), $i$-th output for $\mathbf{O}$ corresponds to the guess that the chosen state was originally $s_i$.
(Without loss of generality, it suffices to consider $N$-valued observables when discriminating $N$ states.)
Our aim is to make the success probability as high as possible.
The \emph{optimal success probability} $P_{\mathrm{succ}}$ is obviously defined by
\begin{equation}
\label{eq:optimal_success_prob}
P_{\mathrm{succ}} = \sup_{\mathbf{O} \in \mathcal{O}_N} P_{\mathrm{succ}}(\mathbf{O}) \enspace,
\end{equation}
and an observable $\mathbf{O} \in \mathcal{O}_N$ is called \emph{optimal} if it attains the supremum, namely: $P_{\mathrm{succ}}(\mathbf{O}) = P_{\mathrm{succ}}$.
However, it is nontrivial whether or not an optimal observable exists in each case.
Ozawa \cite{Oza80} has proven existence of Bayes optimal measurements under somewhat different formulation.
The existence theorem also holds in our situation.
Here we present the theorem together with its proof that is significantly simpler than the one in \cite{Oza80}, as follows:
\begin{theorem}
\label{thm:optimal_observable_exist}
The supremum in the right-hand side of \eqref{eq:optimal_success_prob} is attained by an observable $\mathbf{O} \in \mathcal{O}_N$.
Hence an optimal observable always exists.
\end{theorem}
The rest of this section is devoted to the proof of Theorem \ref{thm:optimal_observable_exist}; note that in this proof, virtual states do not appear at all.
The outline is the following: With respect to a certain topology, the set $\mathcal{O}_N$ of $N$-valued observables is compact and the map $\mathcal{O}_N \to \mathbb{R}$, $\mathbf{O} \mapsto P_{\mathrm{succ}}(\mathbf{O})$, is continuous, therefore this map takes the maximum value at some $\mathbf{O} \in \mathcal{O}_N$.
Now we introduce a map $\iota:\mathcal{E} \to \left[0,1\right]^{\mathcal{S}}$ from $\mathcal{E}$ to the direct product $\left[0,1\right]^{\mathcal{S}} = \prod_{s \in \mathcal{S}} \left[0,1\right]_s$ of copies $\left[0,1\right]_s$ of the unit interval $\left[0,1\right]$ over all $s \in \mathcal{S}$ by $\iota(e) = (e(s))_{s \in \mathcal{S}}$ for any $e \in \mathcal{E}$.
Then $\iota$ is injective, therefore $\mathcal{E}$ is identified with the topological subspace $\iota(\mathcal{E})$ of the product space $\left[0,1\right]^{\mathcal{S}}$.
By the definition of product topology, $\mathcal{T}(\left[0,1\right]^{\mathcal{S}})$ is the weakest topology to make every projection $\pi_s:\left[0,1\right]^{\mathcal{S}} \to \left[0,1\right]_s$ ($s \in \mathcal{S}$) continuous.
Thus the topology on $\mathcal{E}$ induced by the identification is the weakest to make every \lq\lq evaluation map'' $\mathsf{ev}_s:\mathcal{E} \to \left[0,1\right]$, $\mathsf{ev}_s(e) = e(s)$ ($s \in \mathcal{S}$) continuous.
Now the following holds:
\begin{lemma}
\label{lem:effect_closed}
$\iota(\mathcal{E})$ is a closed subset of $\left[0,1\right]^{\mathcal{S}}$.
\end{lemma}
\begin{proof}
For each $s,t \in \mathcal{S}$ and $0 \leq \lambda \leq 1$, put $s' = \lambda s + (1 - \lambda) t \in \mathcal{S}$, and let
\begin{displaymath}
\mathcal{A}_{s,t,\lambda} = \{f \in \left[0,1\right]^{\mathcal{S}} \mid \pi_{s'}(f) - \lambda \pi_{s}(f) - (1-\lambda) \pi_{t}(f) = 0\} \enspace.
\end{displaymath}
Then $\mathcal{A}_{s,t,\lambda}$ is a closed subset of $\left[0,1\right]^{\mathcal{S}}$, since the function $\pi_{s'} - \lambda \pi_{s} - (1-\lambda) \pi_{t}$ on $\left[0,1\right]^{\mathcal{S}}$ is continuous.
Moreover, the affine property of the effects implies that $\iota(\mathcal{E})$ is the intersection of all the subsets $\mathcal{A}_{s,t,\lambda}$.
Hence $\iota(\mathcal{E})$ is also closed in $\left[0,1\right]^{\mathcal{S}}$, therefore Lemma \ref{lem:effect_closed} holds.
\end{proof}
By Tychonoff's Theorem, the product space $\left[0,1\right]^{\mathcal{S}}$ is compact, therefore $\mathcal{E}$ is also compact with respect to the above topology by Lemma \ref{lem:effect_closed}.
Thus the product space $\mathcal{E}^N$ is also compact owing to Tychonoff's Theorem again.
Moreover, a similar argument implies that the subset $\mathcal{O}_N$ of $\mathcal{E}^N$ is closed in $\mathcal{E}^N$, since the map $\mathcal{E}^N \to \mathbb{R}$, $(e_i)_{i = 1}^{N} \mapsto \sum_{i = 1}^{N} e_i(s)$, is continuous for every $s \in \mathcal{S}$.
Thus $\mathcal{O}_N$ is also compact.
Finally, with respect to the topology on $\mathcal{O}_N$, the above function $\mathbf{O} \mapsto P_{\mathrm{succ}}(\mathbf{O})$ on $\mathcal{O}_N$ is continuous.
Hence the proof of Theorem \ref{thm:optimal_observable_exist} is concluded.

\section{Helstrom Families}
\label{sec:Helstrom}
In Sect.\ \ref{sec:discrimination_problem}, we have seen that an optimal observable to discriminate given states always exists in general probabilistic theories.
In the quantum cases, the state discrimination problem has been intently investigated (e.g., \cite{BKMH97,Hel_book,JRF02,YKL75}), but strategies for attaining optimal solutions have been well established only in restricted cases such as two-state cases (cf., \cite{Hel_book}) and some symmetric cases (cf., \cite{BKMH97}).
To study this problem in general probabilistic theories, our preceding work \cite{KMI08} introduced and studied the notion of \lq\lq (weak) Helstrom families'' from a geometric viewpoint.
In this section, we give a translation of the preceding formulation to our minimal framework.

Recall that we are given $N$ states $s_i \in \mathcal{S}$ with a priori probabilities $p_i > 0$, $\sum_i p_i = 1$.
Then the definition of weak Helstrom families is the following (cf., Definition 1 in \cite{KMI08}):
\begin{definition}
\label{defn:weakHF}
We call a family of $N$ ensembles $(\widetilde{p}_i,s_i; 1 - \widetilde{p}_i,t_i)$, $i = 1,\dots,N$, a \emph{weak Helstrom family}, if there exist a quantity $p \geq \max_i p_i$ called a \emph{Helstrom ratio}, $N$ real or virtual states $t_i \in \widetilde{\mathcal{S}}$, $i = 1,\dots,N$ called \emph{conjugate states} to $s_i$, and a real or virtual state $s \in \widetilde{\mathcal{S}}$ called a \emph{reference state}, such that 
\begin{equation}
\label{eq:reference_state}
\widetilde{p}_i s_i + (1 - \widetilde{p}_i) t_i = s \enspace,\enspace  \mbox{ with } 0 < \widetilde{p}_i = \frac{p_i}{p} \leq 1
\end{equation}
for every $i$.
We call a weak Helstrom family \emph{trivial} when $p \geq 1$, and \emph{nontrivial} when $p < 1$.
\end{definition}
\begin{example}
\label{exmp:WHF}
In Fig.\ \ref{fig:WHF}, we consider the case $N = 3$ and $p_i = 1/3$ ($i = 1,2,3$).
The three states $t_1,t_2,t_3$ are in such positions that their configuration is similar to that of $s_1,s_2,s_3$ with respect to the center $s$ of similarity, with similarity ratio $\overline{t_i s} / \overline{s_i s} = 2 / 1$.
Now these form a weak Helstrom family with $\widetilde{p}_i = 2/3$, therefore the Helstrom ratio is $p = p_i / \widetilde{p}_i = 1/2$.
Note that any other similar configuration with a larger similarity ratio gives a weak Helstrom family with larger $\widetilde{p}_i$, hence with a smaller Helstrom ratio.
\end{example}
\begin{figure}[hbt]
\centering
\includegraphics[scale=0.25]{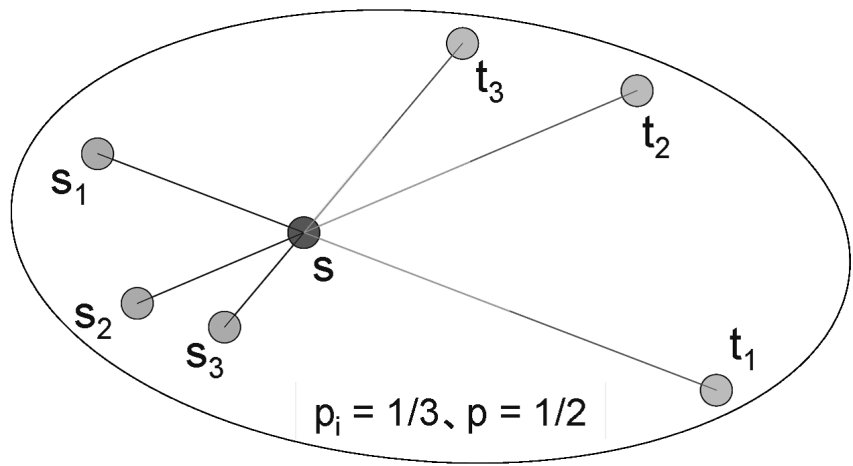}
\caption{Example of weak Helstrom family}
\label{fig:WHF}
\end{figure}
In the original paper \cite{KMI08}, a weak Helstrom family was required to satisfy an additional condition $p \leq 1$, but here we relax this condition to simplify the argument.
Note that a trivial weak Helstrom family with Helstrom ratio $p = 1$ always exists, by taking conjugate states $t_i = (1 - p_i)^{-1} \sum_{j \neq i} p_j s_j$ and a reference state $s = \sum_{i} p_i s_i$.
Example \ref{exmp:WHF} suggests that, intuitively, some nontrivial weak Helstrom families can be found as well by taking the states $t_i$ with larger configuration (cf., \cite{KMI08}).
An importance of weak Helstrom families in a study of state discrimination problems is implied by the following property that has been proven in \cite{KMI08} under the framework there:
\begin{proposition}
[cf., Proposition 1 in {\cite{KMI08}}]
\label{prop:weakHF_bound}
For any weak Helstrom family with Helstrom ratio $p$, we have $P_{\mathrm{succ}} \leq p$ for the optimal success probability.
\end{proposition}
\begin{proof}
The idea of proof is essentially the same as \cite{KMI08}.
For any observable $\mathbf{O} = (e_i)_i \in \mathcal{O}_N$ with the corresponding virtual observable $\widetilde{\mathbf{O}} = (\widetilde{e_i})_i \in \widetilde{\mathcal{O}}_N$ (see Lemma \ref{lem:observable_correspondence}), we have
\begin{eqnarray*}
1
&=& \sum_{i = 1}^{N} \widetilde{e_i}(s) \qquad \mbox{ (using $\sum_i \widetilde{e_i} = 1$)} \\
&=& \sum_i \widetilde{e_i}( \widetilde{p}_i s_i + (1 - \widetilde{p}_i) t_i ) \qquad \mbox{ (using \eqref{eq:reference_state})} \\
&=& \sum_i \widetilde{p}_i e_i(s_i) + \sum_i (1 - \widetilde{p}_i) \widetilde{e_i}(t_i) \qquad \mbox{ (using $s_i \in \mathcal{S}$)} \\
&=& \sum_i \frac{ p_i }{ p } e_i(s_i) + \sum_i (1 - \widetilde{p}_i) \widetilde{e_i}(t_i) \qquad \mbox{ (using \eqref{eq:reference_state})} \\
&=& \frac{ P_{\mathrm{succ}}(\mathbf{O}) }{ p } + \sum_i (1 - \widetilde{p}_i) \widetilde{e_i}(t_i) \qquad \mbox{ (using \eqref{eq:success_prob}).}
\end{eqnarray*}
Since $\widetilde{p}_i \leq 1$, the second term of the last row is nonnegative, therefore we have $P_{\mathrm{succ}}(\mathbf{O}) \leq p$ for any $\mathbf{O} \in \mathcal{O}_N$.
Hence Proposition \ref{prop:weakHF_bound} holds.
\end{proof}
Note that the bound $P_{\mathrm{succ}} \leq p$ given by Proposition \ref{prop:weakHF_bound} is meaningless when the weak Helstrom family is trivial.
Thus only the weak Helstrom families that are significant for our purpose are the nontrivial ones.
Now it was mentioned in Example \ref{exmp:WHF} that changing the configuration of conjugate states to larger one makes the Helstrom ratio smaller, hence makes the bound given by Proposition \ref{prop:weakHF_bound} closer to the tight one.
We are interested in whether or not the tight bound can be achieved just by this strategy.
Owing to the observation, a notion of \lq\lq Helstrom families'', that is a special subclass consisting of \lq\lq optimal'' weak Helstrom families, was introduced in \cite{KMI08}:
\begin{definition}
[{cf., Definition 2 in \cite{KMI08}}]
\label{defn:HF}
We call a weak Helstrom family a \emph{Helstrom family} if the Helstrom ratio $p$ attains the tight bound: $p = P_{\mathrm{succ}}$.
\end{definition}
If a Helstrom family exists, then we can determine (at least in principle) the optimal success probability by only searching (weak) Helstrom families by a certain (for example, geometric) method.
However, existence of Helstrom families has been proven in the original work \cite{KMI08} only for some restricted cases.
In this article, we investigate existence of Helstrom families in more general situations.

For this purpose, it is worthy to study conditions for a weak Helstrom family to be a Helstrom family.
For one direction, a sufficient condition has been given in \cite{KMI08}.
Here we prove the same result under the present framework:
\begin{proposition}
[{cf., Proposition 2 in \cite{KMI08}}]
\label{prop:sufficient_HF}
A sufficient condition for a weak Helstrom family $(\widetilde{p}_i,s_i; 1 - \widetilde{p}_i,t_i)$, $i = 1,\dots,N$, to be a Helstrom family is that there exists $\mathbf{O} = (e_i)_{i = 1}^{N} \in \mathcal{O}_N$ such that $\widetilde{e_i}(t_i) = 0$ for every $i$.
Moreover, such an observable $\mathbf{O}$ is optimal (if exists): $P_{\mathrm{succ}}(\mathbf{O}) = P_{\mathrm{succ}}$.
\end{proposition}
\begin{proof}
The idea is again the same as \cite{KMI08}.
For such an observable $\mathbf{O}$, the argument in the proof of Proposition \ref{prop:weakHF_bound} implies that
\begin{displaymath}
1 = \frac{ P_{\mathrm{succ}}(\mathbf{O}) }{ p } + \sum_i (1 - \widetilde{p}_i) \widetilde{e_i}(t_i) = \frac{ P_{\mathrm{succ}}(\mathbf{O}) }{ p } \enspace,
\end{displaymath}
hence $P_{\mathrm{succ}}(\mathbf{O}) = p$.
Now we have $P_{\mathrm{succ}}(\mathbf{O}) \leq P_{\mathrm{succ}} \leq p = P_{\mathrm{succ}}(\mathbf{O})$ by Proposition \ref{prop:weakHF_bound}, therefore $P_{\mathrm{succ}}(\mathbf{O}) = P_{\mathrm{succ}} = p$.
Hence Proposition \ref{prop:sufficient_HF} holds.
\end{proof}
Again, Helstrom families are closely related to optimal state discrimination via Proposition \ref{prop:sufficient_HF}.
In the special case of two-state discrimination (i.e., $N = 2$), the above condition is rephrased as follows.
Here we use the following terminology:
\begin{definition}
\label{defn:distinguishable}
Two real or virtual states $t_1,t_2 \in \widetilde{\mathcal{S}}$ are said to be \emph{distinguishable} if there exists an $\widetilde{e} \in \widetilde{\mathcal{E}}$ such that $\widetilde{e}(t_1) = 1$ and $\widetilde{e}(t_2) = 0$, i.e., the virtual observable $(1 - \widetilde{e},\widetilde{e}) \in \widetilde{\mathcal{O}}_2$ discriminates $t_1$ and $t_2$ \emph{with certainty}. 
\end{definition}
Then the rephrased condition is the following:
\begin{corollary}
[cf., {Theorem 1 in \cite{KMI08}}]
\label{cor:sufficient_HF_two}
Let $(\widetilde{p}_i,s_i; 1 - \widetilde{p}_i,t_i)$, $i = 1,2$, be a weak Helstrom family for two states $s_1,s_2 \in \mathcal{S}$ with a priori probabilities $p_1,p_2$.
If the conjugate states $t_1$ and $t_2$ are distinguishable, then this weak Helstrom family is a Helstrom family.
Moreover, an optimal observable $\mathbf{O}$ is given by an effect $e$ with the corresponding virtual effect $\widetilde{e}$ distinguishing $t_1$ and $t_2$: $\mathbf{O} = (1-e,e)$.
\end{corollary}
Now owing to the existence of an optimal observable (Theorem \ref{thm:optimal_observable_exist}), we obtain a \lq\lq converse'' of the above facts.
To state the result precisely, we recall the following notion introduced in \cite{KMI08}:
\begin{definition}
[{\cite{KMI08}}]
\label{defn:generic}
By \emph{generic case} we signify any case in which the optimal success probability satisfies $P_{\mathrm{succ}} > \max_i p_i$, and by \emph{non-generic case} we signify any of the remaining cases, i.e., $P_{\mathrm{succ}} = \max_i p_i$.
\end{definition}
This definition means that, in non-generic cases, an optimal observable is always given by the trivial one that always returns $i$-th output with the index $i$ determined by $p_i = \max_j p_j$; namely, we always guess that a given state would be the most frequent $s_i$.
Hence the state discrimination problem is nontrivial only in generic cases.
Now we give the following result stating that the sufficient condition in Proposition \ref{prop:sufficient_HF} is also necessary in generic cases:
\begin{proposition}
\label{prop:necessary_HF}
Let $(\widetilde{p}_i,s_i; 1 - \widetilde{p}_i,t_i)$, $i = 1,\dots,N$, be a Helstrom family.
Then, in generic cases, an optimal observable $\mathbf{O} = (e_i)_i \in \mathcal{O}_N$ for discriminating given states satisfies $\widetilde{e_i}(t_i) = 0$ for every $i$.
\end{proposition}
\begin{proof}
For any optimal observable $\mathbf{O} = (e_i)_i$, since $P_{\mathrm{succ}}(\mathbf{O}) = P_{\mathrm{succ}} = p$, the argument in Proposition \ref{prop:weakHF_bound} implies that
\begin{displaymath}
1 = \frac{ P_{\mathrm{succ}}(\mathbf{O}) }{ p } + \sum_i (1 - \widetilde{p}_i) \widetilde{e_i}(t_i) = 1 + \sum_i (1 - \widetilde{p}_i) \widetilde{e_i}(t_i) \enspace,
\end{displaymath}
therefore $\sum_{i} (1 - \widetilde{p}_i) \widetilde{e_i}(t_i) = 0$.
Thus we have either $\widetilde{p}_i = 1$ for some $i$, or $\widetilde{e_i}(t_i) = 0$ for every $i$.
Now if $\widetilde{p}_i = 1$, then $P_{\mathrm{succ}} = p = p_i / \widetilde{p}_i = p_i$, contradicting the assumption that we are in a generic case.
Hence Proposition \ref{prop:necessary_HF} holds.
\end{proof}
\begin{corollary}
\label{cor:necessary_HF_two}
Let $(\widetilde{p}_i,s_i; 1 - \widetilde{p}_i,t_i)$, $i = 1,2$, be a Helstrom family for two states $s_1,s_2$ with a priori probabilities $p_1,p_2$.
Then, in generic cases, the conjugate states $t_1$ and $t_2$ are distinguishable by a virtual effect $\widetilde{e} \in \widetilde{\mathcal{E}}$ corresponding to an optimal observable $\mathbf{O} = (1-e,e) \in \mathcal{O}_2$ for discriminating the states $s_1$ and $s_2$.
\end{corollary}
\begin{proof}
By Proposition \ref{prop:necessary_HF}, an optimal observable $\mathbf{O} = (1-e,e) \in \mathcal{O}_2$ satisfies that $(1-\widetilde{e})(t_1) = 0$ and $\widetilde{e}(t_2) = 0$, therefore $\widetilde{e}(t_1) = 1$.
\end{proof}

\section{Existence of Helstrom Families for Two-State Cases}
\label{sec:existence}

In Sect.\ \ref{sec:Helstrom}, we have presented some properties of (weak) Helstrom families for $N$-state cases.
However, existence of Helstrom families has not been clarified so far.
In this section, we investigate existence of Helstrom families particularly in two-state cases, i.e., $N = 2$.
Note that our argument in this section works in a general setting, \emph{not} necessarily classical or quantum, and also is \emph{not} restricted to finite-dimensional cases.

Throughout this section, fix states $s_1,s_2 \in \mathcal{S}$ and a priori probabilities $p_1,p_2$.
For simplicity, we assume that $s_1 \neq s_2$ and $p_1 \geq p_2$ by symmetry.
Any (weak) Helstrom family in this section is for $s_1,s_2$ and $p_1,p_2$ unless otherwise specified.

\subsection{A condition for generic cases}
\label{subsec:HF_condition_generic}
In this subsection, we present a condition for generic cases for later use.
First we introduce an element $s^{\ast}$ of $V$ that plays a significant role in our following argument.
Recall that we have assumed $p_1 \geq p_2$.
If $p_1 > p_2$, then define
\begin{displaymath}
s^{\ast} = \frac{ p_1 s_1 - p_2 s_2 }{ p_1 - p_2 }
= s_1 + \frac{ p_1 }{ p_1 - p_2 } (s_1 - s_2) \enspace.
\end{displaymath}
Note that $s^{\ast} \in V$ since $s_1$ and $s_2$ are real states, therefore we have $s^{\ast} \in \widetilde{\mathcal{S}}$ if and only if $s^{\ast} \in \mathrm{cl}_V(\mathcal{S})$.
Then the aforementioned condition is the following:
\begin{lemma}
\label{lem:HFandGeneric}
\begin{enumerate}
\item If the following condition
\begin{equation}
\label{eq:assumption_generic}
\mbox{either } p_1 = p_2 \mbox{, or } p_1 > p_2 \mbox{ and } s^{\ast} \not\in \widetilde{\mathcal{S}}
\end{equation}
is satisfied and a Helstrom family exists, then it is a generic case.
\item If $p_1 > p_2$ and $s^{\ast} \in \widetilde{\mathcal{S}}$, then it is a non-generic case.
\end{enumerate}
\end{lemma}
\begin{proof}
Note that for any Helstrom family $(\widetilde{p}_i,s_i; 1 - \widetilde{p}_i,t_i)$, $i = 1,2$, it is a non-generic case if and only if $\widetilde{p}_1 = 1$ (since $p_1 \geq p_2$).
Now if $p_1 = p_2$ and a Helstrom family exists, then $\widetilde{p}_1 = 1$ implies that $\widetilde{p}_2 = 1$ and $s = s_1 = s_2$ (see \eqref{eq:reference_state}), contradicting the assumption $s_1 \neq s_2$.
If $p_1 > p_2$, $s^{\ast} \not\in \widetilde{\mathcal{S}}$ and a Helstrom family exists, then $\widetilde{p}_1 = 1$ implies that $p_1 = p = p_2 / \widetilde{p}_2$, $s = s_1 = \widetilde{p}_2 s_2 + (1 - \widetilde{p}_2) t_2$ and
\begin{displaymath}
t_2 = \frac{ s_1 - \widetilde{p}_2 s_2 }{ 1 - \widetilde{p}_2 }
= \frac{ p_1 s_1 - p_1 \widetilde{p}_2 s_2 }{ p_1 - p_1 \widetilde{p}_2 }
= s^{\ast} \not\in \widetilde{\mathcal{S}} \enspace,
\end{displaymath}
a contradiction.
Thus the first part of the lemma hold.
For the second part, if $p_1 > p_2$ and $s^{\ast} \in \widetilde{\mathcal{S}}$, then we have $s_1 = (1 - p_2 / p_1) s^{\ast} + (p_2 / p_1) s_2$ by the definition of $s^{\ast}$, therefore for any $\mathbf{O} = (1-e,e) \in \mathcal{O}_2$ we have
\begin{eqnarray*}
P_{\mathrm{succ}}(\mathbf{O})
&=& p_1 (1 - e(s_1)) + p_2 e(s_2) \\
&=& p_1 - p_1 \left( \left( 1 - \frac{ p_2 }{ p_1 } \right) \widetilde{e}(s^{\ast}) + \frac{ p_2 }{ p_1 } e(s_2) \right) + p_2 e(s_2) \\
&=& p_1 - (p_1 - p_2) \widetilde{e}(s^{\ast}) \enspace.
\end{eqnarray*}
Since $s^{\ast} \in \widetilde{\mathcal{S}}$, we have $\widetilde{e}(s^{\ast}) \geq 0$, therefore $P_{\mathrm{succ}}(\mathbf{O}) \leq p_1$ for any $\mathbf{O} \in \mathcal{O}_2$.
This means that it is a non-generic case.
Hence Lemma \ref{lem:HFandGeneric} holds.
\end{proof}
Owing to this lemma, in what follows we assume that the condition \eqref{eq:assumption_generic} in Lemma \ref{lem:HFandGeneric} is satisfied unless otherwise specified, in order to focus on generic cases.
In the following subsections we will prove that a Helstrom family always exists under the assumption \eqref{eq:assumption_generic}, that is our main result in this article.

\subsection{Auxiliary results}
\label{subsec:HF_auxiliary_result}
In this subsection, for later use we summarize some known facts for topological vector spaces, together with some further properties.
Our main reference is the book \cite{SW_book}.
See also Sect.\ \ref{sec:formulation_model} for terminology.

First we list the following (special cases of the) facts presented in \cite{SW_book}:
\begin{theorem}
[{Theorem 3.2 in \cite[Chap.\ I]{SW_book}}]
\label{thm:fin_dim_is_Euclidean}
Any $n$-dimensional Hausdorff t.v.s.\ with $n < \infty$ is isomorphic to the $n$-dimensional Euclidean space $\mathbb{R}^n$.
\end{theorem}
\begin{proposition}
[{Proposition 3.3 in \cite[Chap.\ I]{SW_book}}]
\label{prop:fin_dim_is_closed}
Let $W$ be a t.v.s.
If $W'$ is a linear subspace of $W$ that is closed in $W$, and $W''$ is a finite-dimensional linear subspace of $W$, then $W' + W''$ is closed in $W$.
\end{proposition}
\begin{proposition}
[{Proposition 3.4 in \cite[Chap.\ I]{SW_book}}]
\label{prop:fin_dim_continuous}
Every linear functional on a finite-dimensional Hausdorff t.v.s.\ is continuous.
\end{proposition}

The next theorem is a variant of Hahn-Banach's Theorem.
Here we use the following notion: A real-valued function $g$ on a vector space $W$ is called a \emph{semi-norm} if we have $g(x + y) \leq g(x) + g(y)$ for any $x,y \in W$ and we have $g(\lambda x) = |\lambda| g(x)$ for any $x \in W$ and $\lambda \in \mathbb{R}$.
Then we have the following theorem:
\begin{theorem}
[{Theorem 3.2 in \cite[Chap.\ II]{SW_book}}]
\label{thm:Hahn-Banach}
Let $W$ be a vector space, $g$ a semi-norm on $W$, and $W'$ a linear subspace of $W$.
If $f$ is a linear functional on $W'$ such that $|f(x)| \leq g(x)$ for all $x \in W'$, then $f$ extends to a linear functional $\overline{f}$ on $W$ such that $|\overline{f}(x)| \leq g(x)$ for all $x \in W$.
\end{theorem}

A subset $C$ of a vector space $W$ is called \emph{circled} if $x \in C$ and $-1 \leq \lambda \leq 1$ imply $\lambda x \in C$; and called \emph{radial} if for any $x \in W$, there exists $\lambda_0 \in \mathbb{R}$ such that $x \in \lambda C$ whenever $|\lambda| \geq |\lambda_0|$.
If $C$ is convex, radial and circled, then the \emph{Minkowski functional} (or \emph{gauge}) $g_C:W \to \mathbb{R}$ of $C$ is defined by
\begin{equation}
\label{eq:Minkowski functional}
g_C(x) = \inf\{\lambda > 0 \mid x \in \lambda C\} \mbox{ for each } x \in W \enspace.
\end{equation}
\begin{proposition}
[{Proposition 1.4 in \cite[Chap.\ II]{SW_book}}]
\label{prop:Minkowski_functional}
The Minkowski functional $g_C$ of $C$ is a semi-norm on $W$.
\end{proposition}

From now, we present the following two properties of our minimal framework (see Theorem \ref{thm:state_space_embedded}) that are consequences of the above facts:
\begin{corollary}
\label{cor:fin_dim_is_Euclidean}
Every finite-dimensional affine subspace $W$ of the t.v.s.\ $\widetilde{V}$ is closed in $\widetilde{V}$ and is isomorphic to the Euclidean space $\mathbb{R}^n$ with $n = \dim W$.
Hence $\widetilde{\mathcal{S}} \cap W$ is a compact subset of $W$.
\end{corollary}
\begin{proof}
The compactness of $\widetilde{\mathcal{S}} \cap W$ follows from the remaining parts.
Since the topology $\mathcal{T}(\widetilde{V})$ of $\widetilde{V}$ is invariant under any translation $x \mapsto x + x_0$, $x_0 \in \widetilde{V}$, we assume without loss of generality that $W$ is a linear subspace of $\widetilde{V}$.
Since $\widetilde{V}$ is Hausdorff, the assertion $W \simeq \mathbb{R}^n$ follows from Theorem \ref{thm:fin_dim_is_Euclidean}; while the null subspace $\{0\}$ of $\widetilde{V}$ is closed in $\widetilde{V}$, therefore $W = \{0\} + W$ is also closed by Proposition \ref{prop:fin_dim_is_closed}.
Hence Corollary \ref{cor:fin_dim_is_Euclidean} holds.
\end{proof}
\begin{corollary}
\label{cor:extension_from_fin_dim}
Let $W$ be a finite-dimensional affine subspace of $\widetilde{V}$.
Then any affine functional $f$ on $W$ extends to a continuous affine functional $\overline{f}$ on $\widetilde{V}$.
\end{corollary}
\begin{proof}
Fix an element $x_0 \in W$ and put $\alpha = f(x_0)$.
Then the linear functional $g:x \mapsto f(x + x_0) - \alpha$ on a finite-dimensional linear subspace $W - x_0$ of $\widetilde{V}$ is continuous by Proposition \ref{prop:fin_dim_continuous} since $\widetilde{V}$ is Hausdorff.
Moreover, since $\widetilde{V}$ is l.c., a consequence of Hahn-Banach's Theorem (Theorem \ref{thm:extension_continuous_to_continuous}) implies that this $g$ extends to a $\overline{g} \in \mathcal{L}_c(\widetilde{V})$.
Now the map $\overline{f}$ defined by $\overline{f}(x) = \overline{g}(x - x_0) + \alpha$ is an affine extension of $f$, and $\overline{f}$ is continuous since the translation $x \mapsto x - x_0$ is an isomorphism from $\widetilde{V}$ to itself.
Hence Corollary \ref{cor:extension_from_fin_dim} holds.
\end{proof}

\subsection{Candidates of conjugate states for Helstrom families}
\label{subsec:HF_conjugate_state}

In this subsection, we investigate the candidates of conjugate states $t_1,t_2$ for Helstrom families.
For the purpose, we introduce some further notations.
Recall that we have assumed the condition \eqref{eq:assumption_generic}.
In the case $p_1 = p_2$, let $\mathcal{C}_{\mathrm{weak}}$ be the set of all pairs $(t_1,t_2)$ of distinct $t_1,t_2 \in \widetilde{\mathcal{S}}$ such that the vector $\overrightarrow{ t_1 t_2 }$ is proportional to $\overrightarrow{ s_2 s_1 }$ (i.e., $t_2 - t_1 = c (s_1 - s_2)$ for some $c > 0$).
On the other hand, in the case $p_1 > p_2$ and $s^{\ast} \not\in \widetilde{\mathcal{S}}$, let $\mathcal{C}_{\mathrm{weak}}$ be the set of all pairs $(t_1,t_2)$ of distinct $t_1,t_2 \in \widetilde{\mathcal{S}}$ such that $t_2$ lies in the line segment $\overline{t_1 s^{\ast}} = \mathrm{Conv}(\{t_1,s^{\ast}\})$ between $t_1$ and $s^{\ast}$.
Note that $\mathcal{C}_{\mathrm{weak}} \neq \emptyset$ since $(s_2,s_1) \in \mathcal{C}_{\mathrm{weak}}$.
The next lemma shows that $\mathcal{C}_{\mathrm{weak}}$ is the set of the pairs of conjugate states for weak Helstrom families:
\begin{lemma}
\label{lem:conjugate_state_weakHF}
If $(\widetilde{p}_i,s_i;1 - \widetilde{p}_i,t_i)$, $i = 1,2$, is a weak Helstrom family, then $(t_1,t_2) \in \mathcal{C}_{\mathrm{weak}}$.
Conversely, if $(t_1,t_2) \in \mathcal{C}_{\mathrm{weak}}$, then there exist $0 < \widetilde{p}_i \leq 1$, $i = 1,2$, such that $(\widetilde{p}_i,s_i;1 - \widetilde{p}_i,t_i)$, $i = 1,2$, is a weak Helstrom family.
\end{lemma}
\begin{proof}
First, we consider the case $p_1 = p_2$.
Then any weak Helstrom family satisfies $\widetilde{p}_1 = \widetilde{p}_2$, therefore \eqref{eq:reference_state} implies that $\widetilde{p}_1 < 1$ (otherwise, we have $s_1 = s = s_2$, contradicting the fact $s_1 \neq s_2$) and $t_2 - t_1 = \widetilde{p}_1 ( s_1 - s_2 ) / (1 - \widetilde{p}_1)$.
Thus $(t_1,t_2) \in \mathcal{C}_{\mathrm{weak}}$.
Conversely, if $(t_1,t_2) \in \mathcal{C}_{\mathrm{weak}}$, then $t_2 - t_1 = c(s_1 - s_2)$ for some $c > 0$, while this $c$ can be written as $c = \widetilde{p} / (1 - \widetilde{p})$ with $0 < \widetilde{p} < 1$.
Now it follows that $(\widetilde{p},s_i; 1 - \widetilde{p},t_i)$, $i = 1,2$, is a weak Helstrom family.
Thus the lemma holds in this case.

Secondly, we consider the case $p_1 > p_2$ and $s^{\ast} \not\in \widetilde{\mathcal{S}}$.
Then by \eqref{eq:reference_state}, any weak Helstrom family satisfies $\widetilde{p}_2 = p_2 / p = \widetilde{p}_1 p_2 / p_1 < \widetilde{p}_1 \leq 1$, therefore
\begin{equation}
\label{eq:conjugate_state_weakHF}
t_2 = \frac{ \widetilde{p}_1 s_1 - \widetilde{p}_2 s_2 + ( 1 - \widetilde{p}_1 ) t_1 }{ 1 - \widetilde{p}_2 }
= \lambda t_1 + (1 - \lambda) s^{\ast} \enspace, \mbox{ where } \lambda = \frac{ 1 - \widetilde{p}_1 }{ 1 - \widetilde{p}_2 } \enspace.
\end{equation}
Now we have $0 \leq \lambda < 1$ since $\widetilde{p}_2 < \widetilde{p}_1 \leq 1$, therefore $t_2 \in \overline{ t_1 s^{\ast} }$.
Moreover, if $t_1 = t_2$, then \eqref{eq:conjugate_state_weakHF} implies that $t_1 = s^{\ast}$, contradicting $t_1 \in \widetilde{\mathcal{S}}$ and $s^{\ast} \not\in \widetilde{\mathcal{S}}$.
Thus $(t_1,t_2) \in \mathcal{C}_{\mathrm{weak}}$.
Conversely, if $(t_1,t_2) \in \mathcal{C}_{\mathrm{weak}}$, then we have $t_2 = \lambda t_1 + (1 - \lambda) s^{\ast}$ for some $0 \leq \lambda < 1$, and now $(\widetilde{p}_i,s_i; 1 - \widetilde{p}_i,t_i)$, $i = 1,2$, is a weak Helstrom family for $\widetilde{p}_1 = (p_1 - p_1 \lambda) / (p_1 - p_2 \lambda)$ and $\widetilde{p}_2 = \widetilde{p}_1 p_2 / p_1$.
Hence Lemma \ref{lem:conjugate_state_weakHF} holds.
\end{proof}
By the lemma and Corollary \ref{cor:sufficient_HF_two}, for finding a Helstrom family, it suffices to search a pair $(t_1,t_2) \in \mathcal{C}_{\mathrm{weak}}$ such that $t_1$ and $t_2$ are distinguishable by a virtual effect $\widetilde{e} \in \widetilde{\mathcal{E}}$ (see Definition \ref{defn:distinguishable} for terminology).
The outline to prove the existence of such a pair $(t_1,t_2)$ is the following:
\begin{enumerate}
\item Define a function $\ell:\mathcal{C}_{\mathrm{weak}}' \to \mathbb{R}$, where
\begin{displaymath}
\mathcal{C}_{\mathrm{weak}}' = \mathcal{C}_{\mathrm{weak}} \cup \{(t,t) \mid t \in \widetilde{\mathcal{S}}\} \subset \widetilde{\mathcal{S}} \times \widetilde{\mathcal{S}} \enspace,
\end{displaymath}
such that $\ell \geq 0$ and $\ell(t_1,t_2) = 0$ if and only if $t_1 = t_2$; hence $\ell > 0$ on $\mathcal{C}_{\mathrm{weak}}$.
\item Prove that $\mathcal{C}_{\mathrm{weak}}'$ is closed in $\widetilde{\mathcal{S}} \times \widetilde{\mathcal{S}}$; hence $\mathcal{C}_{\mathrm{weak}}'$ is compact since $\widetilde{\mathcal{S}} \times \widetilde{\mathcal{S}}$ is.
\item Prove that $\ell$ is continuous; hence $\ell$ takes the maximum value at some pair $(t_1,t_2)$ in $\mathcal{C}_{\mathrm{weak}}$ (see the first step).
\item Prove that $t_1$ and $t_2$ are distinguishable.
\end{enumerate}
From now, we proceed the program.
In what follows, for a t.v.s.\ $W$, let $\mathcal{L}(W)$, $\mathcal{L}_c(W)$, $\mathcal{A}(W)$, and $\mathcal{A}_c(W)$ denote, respectively, the sets of linear functionals on $W$, of continuous linear functionals on $W$, of affine functionals on $W$, and of continuous affine functionals on $W$.

For the first step of the program, we define the function $\ell:\mathcal{C}_{\mathrm{weak}}' \to \mathbb{R}$ as follows: In the case $p_1 = p_2$, define $\ell(t_1,t_2)$ by 
\begin{displaymath}
t_2 - t_1 = \ell(t_1,t_2) (s_1 - s_2) \mbox{ for } (t_1,t_2) \in \mathcal{C}_{\mathrm{weak}}' \enspace.
\end{displaymath}
On the other hand, in the case $p_1 > p_2$ and $s^{\ast} \not\in \widetilde{\mathcal{S}}$, define $\ell(t_1,t_2)$ by
\begin{displaymath}
t_2 = \ell(t_1,t_2) s^{\ast} + (1 - \ell(t_1,t_2)) t_1 \mbox{ for } (t_1,t_2) \in \mathcal{C}_{\mathrm{weak}}'
\end{displaymath}
(thus $0 \leq \ell < 1$; note that $\ell \neq 1$ since $t_2 \neq s^{\ast}$).
This $\ell$ has the properties specified in the first step.
Note that $\ell(t_1,t_2)$ becomes larger if and only if $t_1$ and $t_2$ become \lq\lq far'' from each other in the space $\widetilde{\mathcal{S}}$ (in an intuitive sense; this becomes a strict sense at least in finite-dimensional cases, since in such a case $\widetilde{\mathcal{S}}$ admits the Euclidean metric); hence our program to make the value $\ell(t_1,t_2)$ as large as possible also fits the strategy mentioned in Example \ref{exmp:WHF} for decreasing the Helstrom ratio.
Namely, in the case $p_1 = p_2$, the definition of $\ell$ intuitively implies that $\ell(t_1,t_2)$ is the \lq\lq distance'' between $t_1$ and $t_2$ normalized as the \lq\lq distance'' between $s_1$ and $s_2$ being $1$.
On the other hand, in the case $p_1 > p_2$ and $s^{\ast} \not\in \widetilde{\mathcal{S}}$, the definition of $\ell$ implies that
\begin{displaymath}
t_1 - t_2 = \frac{ \ell(t_1,t_2) }{ 1 - \ell(t_1,t_2) } (t_2 - s^{\ast}) \enspace,
\end{displaymath}
therefore $\ell(t_1,t_2) / (1 - \ell(t_1,t_2))$, that is increasing for $\ell(t_1,t_2)$, is the \lq\lq distance'' between $t_1$ and $t_2$ normalized as the \lq\lq distance'' between $t_2$ and $s^{\ast}$ being $1$.

For the second step, we have the following result:
\begin{lemma}
\label{lem:pairs_closed_set}
Let $(t_1,t_2) \in \widetilde{\mathcal{S}} \times \widetilde{\mathcal{S}}$.
\begin{enumerate}
\item If $p_1 = p_2$, then we have $(t_1,t_2) \in \mathcal{C}_{\mathrm{weak}}'$ if and only if $\widetilde{e}(t_1) \geq \widetilde{e}(t_2)$ for any $e \in \mathcal{E}$ such that $e(s_1) \leq e(s_2)$.
\item If $p_1 > p_2$, then we have $(t_1,t_2) \in \mathcal{C}_{\mathrm{weak}}'$ if and only if $f(t_1) \leq f(t_2) \leq f(s^{\ast})$ or $f(t_1) \geq f(t_2) \geq f(s^{\ast})$ for any $f \in \mathcal{A}_c(\widetilde{V})$ such that $f|_{\widetilde{\mathcal{S}}} \in \widetilde{\mathcal{E}}$.
\end{enumerate}
\end{lemma}
\begin{proof}
Since the case $t_1 = t_2$ is trivial, we assume from now that $t_1 \neq t_2$.

For the first part, if $(t_1,t_2) \in \mathcal{C}_{\mathrm{weak}}$, then Lemma \ref{lem:conjugate_state_weakHF} implies that
\begin{displaymath}
s = \widetilde{p} s_1 + (1 - \widetilde{p}) t_1 = \widetilde{p} s_2 + (1 - \widetilde{p}) t_2 \mbox{ for some } 0 < \widetilde{p} < 1 \mbox{ and } s \in \widetilde{\mathcal{S}}
\end{displaymath}
(recall that $s_1 \neq s_2$).
Now for any $e \in \mathcal{E}$, we have
\begin{displaymath}
\widetilde{e}(s) = \widetilde{p} e(s_1) + (1 - \widetilde{p}) \widetilde{e}(t_1) = \widetilde{p} e(s_2) + (1 - \widetilde{p}) \widetilde{e}(t_2) \enspace,
\end{displaymath}
therefore $\widetilde{e}(t_1) \geq \widetilde{e}(t_2)$ whenever $e(s_1) \leq e(s_2)$.
On the other hand, if $(t_1,t_2) \not\in \mathcal{C}_{\mathrm{weak}}'$, then we have either $(t_2,t_1) \in \mathcal{C}_{\mathrm{weak}}$, or $\overrightarrow{ t_1 t_2 } = t_2 - t_1$ is not parallel to the line $\mathrm{Aff}(\{s_1,s_2\})$ containing $s_1$ and $s_2$.
In the former case, we have $e(s_1) < e(s_2)$ for some $e \in \mathcal{E}$ since $\mathcal{S}$ is separated (note that $1 - e \in \mathcal{E}$ and $1 - e(s_1) < 1 - e(s_2)$ whenever $e \in \mathcal{E}$ and $e(s_1) > e(e_2)$), therefore we have $\widetilde{e}(t_2) > \widetilde{e}(t_1)$ in the same way as above.
In the latter case, it is easy to show that $f(s_1) = f(s_2)$ and $f(t_1) < f(t_2)$ for an affine functional $f$ on the affine hull of $\{s_1,s_2,t_1,t_2\}$, and Corollary \ref{cor:extension_from_fin_dim} implies that this $f$ extends to an $\overline{f} \in \mathcal{A}_c(\widetilde{V})$.
Now $\overline{f}(\widetilde{\mathcal{S}})$ is bounded in $\mathbb{R}$ since $\widetilde{\mathcal{S}}$ is compact.
Thus by taking $\alpha > 0$ and $\beta \in \mathbb{R}$ appropriately, the continuous affine functional $g = \alpha \overline{f} + \beta$ on $\widetilde{V}$ satisfies that $g(s_1) = g(s_2)$, $g(t_1) < g(t_2)$ and $g(\widetilde{\mathcal{S}}) \subset \left[0,1\right]$, therefore $\widetilde{e} = g|_{\widetilde{\mathcal{S}}}$ is a virtual effect satisfying $e(s_1) = e(s_2)$ and $\widetilde{e}(t_1) < \widetilde{e}(t_2)$.
Thus the first part of Lemma \ref{lem:pairs_closed_set} holds.

For the second part, note that $s^{\ast} \not\in \widetilde{\mathcal{S}}$ by the assumption \eqref{eq:assumption_generic}.
The \lq\lq only if'' part is now trivial by the definition of $\mathcal{C}_{\mathrm{weak}}$.
To prove the \lq\lq if'' part, assume that $(t_1,t_2) \not\in \mathcal{C}_{\mathrm{weak}}'$.
Then $t_1 \neq t_2$, and we have either $t_1 \in \overline{ t_2 s^{\ast} }$, or $\overrightarrow{ t_1 t_2 }$ is not parallel to the line $\mathrm{Aff}(\{t_1,s^{\ast}\})$ (note that $s^{\ast} \not\in \overline{t_1 t_2}$ since $s^{\ast} \not\in \widetilde{\mathcal{S}}$).
In the former case, since $\widetilde{V}$ is Hausdorff, there exists an $f \in \mathcal{L}_c(\widetilde{V})$ such that $f(t_2) < f(t_1)$.
Now since $\widetilde{\mathcal{S}}$ is compact, an appropriate transformation $g = \alpha f + \beta$ with $\alpha,\beta \in \mathbb{R}$ satisfies that $g(\widetilde{\mathcal{S}}) \subset \left[0,1\right]$ (hence $g|_{\widetilde{\mathcal{S}}} \in \widetilde{\mathcal{E}}$) and $g(t_2) < g(t_1)$, therefore $g(t_1) < g(s^{\ast})$ since $t \in \overline{ t_2 s^{\ast} }$ and $t_1 \neq s^{\ast}$.
On the other hand, in the latter case, we have $f(t_1) = f(s^{\ast}) < f(t_2)$ for an affine functional $f$ on the affine hull of $\{t_1,t_2,s^{\ast}\}$, and Corollary \ref{cor:extension_from_fin_dim} implies that this $f$ extends to an $\overline{f} \in \mathcal{A}_c(\widetilde{V})$.
Now $\overline{f}(\widetilde{\mathcal{S}})$ is bounded in $\mathbb{R}$ since $\widetilde{\mathcal{S}}$ is compact.
Thus by taking an appropriate affine transformation of $\overline{f}$ in the same way as above, it follows that $g(t_1) = g(s^{\ast}) < g(t_2)$ for a $g \in \mathcal{A}_c(\widetilde{V})$ such that $g|_{\widetilde{\mathcal{S}}} \in \widetilde{\mathcal{E}}$.
Hence the second part of Lemma \ref{lem:pairs_closed_set} holds, concluding the proof of Lemma \ref{lem:pairs_closed_set}.
\end{proof}
By this lemma, $\mathcal{C}_{\mathrm{weak}}'$ is closed in $\widetilde{\mathcal{S}} \times \widetilde{\mathcal{S}}$ as desired, since the virtual effect $\widetilde{e} \in \widetilde{\mathcal{E}}$ corresponding to each $e \in \mathcal{E}$ is continuous on $\widetilde{\mathcal{S}}$.

For the third step, we have the following result:
\begin{lemma}
\label{lem:ell_continuous}
The function $\ell$ on $\mathcal{C}_{\mathrm{weak}}'$ is continuous.
\end{lemma}
\begin{proof}
First, we consider the case $p_1 = p_2$.
Fix $e \in \mathcal{E}$ such that $e(s_1) > e(s_2)$ (this is possible since $\mathcal{S}$ is separated), and put $c = (e(s_1) - e(s_2))^{-1} > 0$.
For any $(t_1,t_2) \in \mathcal{C}_{\mathrm{weak}}'$, Lemma \ref{lem:conjugate_state_weakHF} implies that there exists a $\widetilde{p} \in \left[0,1\right)$ such that $\widetilde{p} s_1 + (1 - \widetilde{p}) t_1 = \widetilde{p} s_2 + (1 - \widetilde{p}) t_2 \in \widetilde{\mathcal{S}}$.
Now we have $\ell(t_1,t_2) = \widetilde{p} / (1 - \widetilde{p})$ and $\widetilde{p} e(s_1) + (1 - \widetilde{p}) \widetilde{e}(t_1) = \widetilde{p} e(s_2) + (1 - \widetilde{p}) \widetilde{e}(t_2)$, therefore 
\begin{displaymath}
\ell(t_1,t_2) = \frac{ \widetilde{e}(t_2) - \widetilde{e}(t_1) }{ e(s_1) - e(s_2) } = c( \widetilde{e}(t_2) - \widetilde{e}(t_1) ) \enspace.
\end{displaymath}
This implies that $\ell$ is continuous, since $\widetilde{e} \in \widetilde{\mathcal{E}}$ is continuous.

Secondly, we consider the case that $p_1 > p_2$ and $s^{\ast} \not\in \widetilde{\mathcal{S}}$.
Let $\mathcal{F}$ be the set of all $f \in \mathcal{A}_c(\widetilde{V})$ such that $f|_{\widetilde{\mathcal{S}}} \in \widetilde{\mathcal{E}}$.
Now for each $f \in \mathcal{F}$, put
\begin{displaymath}
A_f = \{(t_1,t_2) \in \widetilde{\mathcal{S}} \times \widetilde{\mathcal{S}} \mid f(t_1) \neq f(s^{\ast})\} \subset \widetilde{\mathcal{S}} \times \widetilde{\mathcal{S}}
\end{displaymath}
and define a function $g_f:A_f \to \mathbb{R}$ by
\begin{displaymath}
g_f(t_1,t_2) = \frac{ f(t_2) - f(t_1) }{ f(s^{\ast}) - f(t_1) } \mbox{ for } (t_1,t_2) \in A_f \enspace.
\end{displaymath}
Since $f$ is continuous, $A_f$ is open in $\widetilde{\mathcal{S}} \times \widetilde{\mathcal{S}}$ and $g_f$ is continuous.
Moreover, we have $\ell(t_1,t_2) = g_f(t_1,t_2)$ for any $(t_1,t_2) \in \mathcal{C}_{\mathrm{weak}}' \cap A_f$ by the definition of $\ell$.
Now we show that
\begin{displaymath}
\ell^{-1}(U) = \bigcup_{f \in \mathcal{F}} (\mathcal{C}_{\mathrm{weak}}' \cap g_f{}^{-1}(U)) \mbox{ for any open subset } U \subset \mathbb{R} \enspace.
\end{displaymath}
Once this is proven, $\ell^{-1}(U)$ is open in $\mathcal{C}_{\mathrm{weak}}'$ since each $g_f{}^{-1}(U) \subset A_f$ is an open subset of $\widetilde{\mathcal{S}} \times \widetilde{\mathcal{S}}$ (recall that $A_f$ is open in $\widetilde{\mathcal{S}} \times \widetilde{\mathcal{S}}$), therefore the continuity of $\ell$ follows.
Since $\ell$ and $g_f$ agree on $\mathcal{C}_{\mathrm{weak}}' \cap A_f$ as above, the inclusion $\supset$ holds immediately.
For the other inclusion, let $(t_1,t_2) \in \mathcal{C}_{\mathrm{weak}}'$ such that $\ell(t_1,t_2) \in U$.
Let $W$ denote the line $\mathrm{Aff}(\{s_1,s_2\})$.
Now if $t_1 \not\in W$, then an argument similar to Lemma \ref{lem:pairs_closed_set} (based on Corollary \ref{cor:extension_from_fin_dim}) implies existence of an $f \in \mathcal{F}$ such that $f$ is constant on $W$ and $f(t_1) \neq f(s_1)$, hence $f(s^{\ast}) = f(s_1) \neq f(t_1)$ (note that $s^{\ast} \in W$).
On the other hand, suppose that $t_1 \in W$.
Since $\widetilde{V}$ is Hausdorff, there exists an $f \in \mathcal{A}_c(\widetilde{V})$ such that $f(s_1) \neq f(s_2)$.
Now by a similar argument as above, this $f$ can be chosen from $\mathcal{F}$.
Since the four points $s_1$, $s_2$, $s^{\ast}$, and $t_1$ are all collinear and $t_1 \neq s^{\ast}$, the fact $f(s_1) \neq f(s_2)$ implies that $f(s^{\ast}) \neq f(t_1)$.
Thus $(t_1,t_2) \in A_f$ in any case, while $\ell$ and $g_f$ agree on $\mathcal{C}_{\mathrm{weak}}' \cap A_f$, therefore $g_f(t_1,t_2) = \ell(t_1,t_2) \in U$ by the above argument.
Hence the inclusion $\subset$ follows, therefore Lemma \ref{lem:ell_continuous} holds.
\end{proof}

For the final part, let $\mathcal{C}$ be the subset of $\mathcal{C}_{\mathrm{weak}}'$ that consists of all pairs in $\mathcal{C}_{\mathrm{weak}}'$ at which $\ell$ takes the maximum value:
\begin{displaymath}
\mathcal{C} = \{(t_1,t_2) \in \mathcal{C}_{\mathrm{weak}}' \mid \ell(t_1,t_2) = \max_{(t'_1,t'_2) \in \mathcal{C}_{\mathrm{weak}}'} \ell(t'_1,t'_2)\} \enspace.
\end{displaymath}
Note that $\emptyset \neq \mathcal{C} \subset \mathcal{C}_{\mathrm{weak}}$ by the above argument.
From now, we show that for any $(t_1,t_2) \in \mathcal{C}_{\mathrm{weak}}$, $t_1$ and $t_2$ are distinguishable if and only if $(t_1,t_2) \in \mathcal{C}$; in particular, a Helstrom family exists.
First, one direction of this assertion is proven as follows:
\begin{proposition}
\label{prop:HEisExtremal}
If $(t_1,t_2) \in \mathcal{C}_{\mathrm{weak}}$, and $t_1$ and $t_2$ are distinguishable, then $(t_1,t_2) \in \mathcal{C}$.
Hence the pair of conjugate states $t_1,t_2$ in any Helstrom family belongs to $\mathcal{C}$.
\end{proposition}
\begin{proof}
The latter part is derived from the combination of the former part, Lemma \ref{lem:conjugate_state_weakHF}, Lemma \ref{lem:HFandGeneric}, and Corollary \ref{cor:necessary_HF_two}.
To prove the former part, assume contrary that $t_1$ and $t_2$ in $\widetilde{\mathcal{S}}$ are distinguishable by a virtual effect $\widetilde{e} \in \widetilde{\mathcal{E}}$ and $(t_1,t_2) \in \mathcal{C}_{\mathrm{weak}}$ but $\ell(t_1,t_2) < \ell(t'_1,t'_2)$ for some $(t'_1,t'_2) \in \mathcal{C}_{\mathrm{weak}}$.
Since $\mathrm{Aff}(\widetilde{\mathcal{S}}) = \widetilde{V}$, this $\widetilde{e}$ extends to an $f \in \mathcal{A}(\widetilde{V})$.
Let $W = \mathrm{Aff}(\{t_1,t_2,t'_1,t'_2\})$.
Then by Corollary \ref{cor:fin_dim_is_Euclidean}, $W$ is isomorphic to a Euclidean space $\mathbb{R}^n$ with $n = \dim W$ and $\mathcal{S}' = \widetilde{\mathcal{S}} \cap W$ is a compact convex subset of $W$.
Since $(t_1,t_2),(t'_1,t'_2) \in \mathcal{C}_{\mathrm{weak}}$, we have $n \leq 2$ by the definition of $\mathcal{C}_{\mathrm{weak}}$.
Now $H_1 = W \cap f^{-1}(1)$ and $H_2 = W \cap f^{-1}(0)$ are parallel supporting hyperplanes of $\mathcal{S}'$ in $W$ at $t_1$ and at $t_2$, respectively, and $\mathcal{S}'$ lies between $H_1$ and $H_2$.
Note that $t_1,t_2,t'_1,t'_2 \in \mathcal{S}'$.

Now in the case $p_1 = p_2$, $\overrightarrow{ t'_1 t'_2 }$ is parallel to $\overrightarrow{ t_1 t_2 }$ since $(t_1,t_2),(t'_1,t'_2) \in \mathcal{C}_{\mathrm{weak}}$.
Thus it is geometrically obvious that $|t'_1 t'_2| \leq |t_1 t_2|$ (where $|x y|$ denotes the distance between $x$ and $y$ in the Euclidean metric on $\mathbb{R}^n$), since two intersecting points of the line $\mathrm{Aff}(\{t'_1,t'_2\})$ with $H_1$ and with $H_2$, respectively, and $t_1$ and $t_2$ form a parallelogram (see Fig.\ \ref{fig:Menelaus}(a)).
This contradicts the assumption $\ell(t'_1,t'_2) > \ell(t_1,t_2)$.

On the other hand, we consider the case that $p_1 > p_2$ and $s^{\ast} \not\in \widetilde{\mathcal{S}}$.
Note that $s^{\ast} \in W$ since $(t_1,t_2) \in \mathcal{C}_{\mathrm{weak}}$.
Then the assumption $\ell(t_1,t_2) < \ell(t'_1,t'_2)$ implies that $|t'_2 t'_1| / |s^{\ast} t'_2| > |t_2 t_1| / |s^{\ast} t_2|$; in particular, neither $t'_1$ nor $t'_2$ lies on the line segment $\overline{ t_1 t_2 }$.
Let $v_1$ and $v_2$ be the intersecting points of the line $\mathrm{Aff}(\{t'_1,t'_2\})$ with $H_1$ and with $H_2$, respectively (see Fig.\ \ref{fig:Menelaus}(b)).
Then we have
\begin{displaymath}
\frac{ |v_2 v_1| }{ |s^{\ast} v_2| } \geq \frac{ |t'_2 t'_1| }{ |s^{\ast} t'_2| } > \frac{ |t_2 t_1| }{ |s^{\ast} t_2| } \enspace.
\end{displaymath}
However, since $H_1$ and $H_2$ are parallel, two triangles $\triangle s^{\ast}v_1 t_1$ and $\triangle s^{\ast}v_2 t_2$ are similar, therefore we have $|v_2 v_1| / |s^{\ast} v_2| = |t_2 t_1| / |s^{\ast} t_2|$, a contradiction.

Thus a contradiction occurs in both cases.
Hence Proposition \ref{prop:HEisExtremal} holds.
\end{proof}
\begin{figure}[htb]
\centering
\begin{picture}(180,100)(0,-100)
\put(10,-20){\line(1,0){65}}
\put(10,-90){\line(1,0){65}}
\put(54,-90){\line(1,5){14}}
\put(54,-90){\circle*{3}}\put(68,-20){\circle*{3}}
\put(36,-65){\line(1,5){4}}
\put(36,-65){\circle*{3}}\put(40,-45){\circle*{3}}
\put(35,-70){\line(-1,-5){2}}\put(32,-85){\line(-1,-5){2}}
\put(41,-40){\line(1,5){2}}\put(44,-25){\line(1,5){2}}
\put(31,-90){\circle*{3}}\put(45,-20){\circle*{3}}
\put(0,-25){$H_2$}\put(0,-95){$H_1$}
\put(50,-100){$t_1$}\put(65,-15){$t_2$}
\put(40,-68){$t'_1$}\put(44,-45){$t'_2$}
\put(68,-20){\line(-4,-1){40}}
\put(28,-30){\line(-3,-5){15}}
\put(13,-55){\line(1,-2){15}}
\put(54,-90){\line(-5,1){26}}
\put(68,-20){\line(2,-1){10}}
\put(54,-90){\line(1,1){15}}
\put(18,-55){$\mathcal{S}'$}
\put(15,-10){(a)}
\put(100,-23){\line(1,0){70}}
\put(100,-90){\line(1,0){70}}
\put(158,-90){\line(0,1){67}}
\put(158,-90){\circle*{3}}\put(158,-23){\circle*{3}}
\put(158,-20){\line(0,1){8}}
\put(158,-8){\line(0,1){8}}
\put(158,-0){\line(-1,-2){7}}
\put(149,-18){\line(-1,-2){7}}
\put(140,-36){\line(-1,-2){14}}
\put(140,-36){\circle*{3}}\put(126,-64){\circle*{3}}
\put(124,-68){\line(-1,-2){6}}
\put(116,-84){\line(-1,-2){6}}
\put(146.5,-23){\circle*{3}}\put(113,-90){\circle*{3}}
\put(158,0){\circle*{3}}
\put(146,-4){$s^{\ast}$}
\put(160,-20){$t_2$}\put(155,-100){$t_1$}
\put(143,-41){$t'_2$}\put(129,-69){$t'_1$}
\put(136,-20){$v_2$}\put(113,-100){$v_1$}
\put(158,-23){\line(-4,-1){53}}
\put(112,-71){\line(-1,5){7}}
\put(112,-71){\line(3,-2){21}}
\put(158,-90){\line(-5,1){25}}
\put(158,-23){\line(2,-1){12}}
\put(158,-90){\line(1,1){12}}
\put(90,-28){$H_2$}\put(90,-95){$H_1$}
\put(110,-45){$\mathcal{S}'$}
\put(115,-10){(b)}
\end{picture}
\caption{The cases (a) $p_1 = p_2$ and (b) $p_1 > p_2$, $s^{\ast} \not\in \widetilde{\mathcal{S}}$ in Proposition \ref{prop:HEisExtremal}}
\label{fig:Menelaus}
\end{figure}
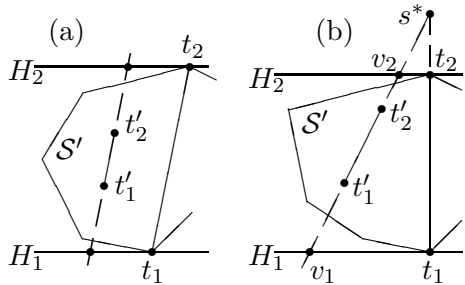

Now we are in a position to state our main theorem in this article, that will be proven in the next subsection:
\begin{theorem}
\label{thm:HE_exist}
If $(t_1,t_2) \in \mathcal{C}$, then $t_1$ and $t_2$ are distinguishable.
Hence, by the above argument, a Helstrom family always exists under the assumption \eqref{eq:assumption_generic}.
\end{theorem}
Before starting the proof we notice the following: Once Theorem \ref{thm:HE_exist} is proven, the hypothesis \lq\lq and a Helstrom family exists'' in the first part of Lemma \ref{lem:HFandGeneric} becomes redundant, therefore the following simple criterion for generic cases in two-state discrimination problems will be obtained that improves Lemma \ref{lem:HFandGeneric}:
\begin{theorem}
\label{thm:criterion_generic}
Under the assumption $p_1 \geq p_2$, the condition \eqref{eq:assumption_generic} is necessary and sufficient for the case to be generic.
\end{theorem}
In particular, an equiprobable case $p_1 = p_2 = 1/2$ is always a generic case, therefore in such a case we can always make a correct guess with probability strictly higher than $1/2$ by using an appropriate observable.

We also mention another nontrivial consequence of Theorem \ref{thm:HE_exist} that shows interesting relations between optimal success probabilities for equiprobable two-state discrimination problems and Gudder's metric functions on the state space (cf., Remark \ref{rem:Gudder_metric}):
\begin{remark}
\label{rem:Gudder_distance_optimal_prob}
First, we translate the definition of Gudder's metric function \cite{Gud_book79} on compact state spaces to our framework with not necessarily compact real state space $\mathcal{S}$.
For $s'_1,s'_2 \in \mathcal{S}$, define $d(s'_1,s'_2)$ to be the infimum of $0 < \lambda \leq 1/2$ such that
\begin{displaymath}
\lambda t_1 + (1 - \lambda) s'_1 = \lambda t_2 + (1 - \lambda) s'_2 \mbox{ for some } t_1,t_2 \in \widetilde{\mathcal{S}}
\end{displaymath}
(note that $\lambda = 1/2$, $t_1 = s'_2$ and $t_2 = s'_1$ always satisfy this condition).
This function $d$ is a metric on $\mathcal{S}$, and this definition coincides with Gudder's original definition in the case $\widetilde{\mathcal{S}} = \mathcal{S}$ (i.e., when $\mathcal{S}$ is compact).
Now the above condition is equivalent to that $(1 - \lambda,s'_i;\lambda,t_i)$, $i = 1,2$, is a weak Helstrom family for states $s'_1,s'_2$ and a priori probabilities $p_i = 1/2$, with Helstrom ratio given by $p = 1 / (2 - 2 \lambda)$.
Thus minimizing $\lambda$ is equivalent to minimizing $p$, and Theorem \ref{thm:HE_exist} implies that the infimum $d(s'_1,s'_2)$ of such $\lambda$ is attained by some Helstrom family, with Helstrom ratio $p = P_{\mathrm{succ}}(s'_1,s'_2)$ where $P_{\mathrm{succ}}(s'_1,s'_2)$ denotes the optimal success probability for discriminating $s'_1$ and $s'_2$ in the equiprobable case.
Thus we have a nontrivial relation
\begin{equation}
\label{eq:relation_prob_metric}
d(s'_1,s'_2) = 1 - \frac{ 1 }{ 2 P_{\mathrm{succ}}(s'_1,s'_2) } \mbox{ for any } s'_1,s'_2 \in \mathcal{S} \enspace.
\end{equation}
In particular, it follows that the function of $s'_1,s'_2$ in the right-hand side is a metric on $\mathcal{S}$.
It seems infeasible to derive the fact just from the intuitive meaning of \lq\lq optimal success probability of state discrimination''.

On the other hand, Gudder also defined another metric function on the same state space, called the \lq\lq intrinsic metric'', by using the former metric function $d$ as a building block.
According to Gudder's definition, we put
\begin{displaymath}
\widetilde{d}(s'_1,s'_2) = \frac{ d(s'_1,s'_2) }{ 1 - d(s'_1,s'_2) } \mbox{ for } s'_1,s'_2 \in \mathcal{S} \enspace.
\end{displaymath}
The concrete structure of the above metric $d$ implies that $\widetilde{d}$ is indeed a metric function and $0 \leq \widetilde{d} \leq 1$.
Moreover, it follows from \eqref{eq:relation_prob_metric} that
\begin{equation}
\label{eq:relation_prob_intrinsic_metric}
\widetilde{d}(s'_1,s'_2) = 2 P_{\mathrm{succ}}(s'_1,s'_2) - 1 \mbox{ for } s'_1,s'_2 \in \mathcal{S} \enspace.
\end{equation}
This shows an operational meaning of Gudder's intrinsic metric that has not been pointed out in the literature.
Moreover, by comparing \eqref{eq:relation_prob_intrinsic_metric} to the well-known formula $P_{\mathrm{succ}}(\rho_1,\rho_2) = 1/2 + D(\rho_1,\rho_2)/2$ for \emph{quantum} states $\rho_1,\rho_2$, where $D(\rho_1,\rho_2)$ denotes the trace distance, Gudder's intrinsic metric coincides with the trace distance for quantum cases.
Hence we have obtained an operationally natural generalization of the trace distance to general probabilistic theories.

Moreover, it is in fact possible to define the \lq\lq trace distance'' in general probabilistic theories directly through the classical trace distance:
\begin{equation}
\label{eq:TD}
D(s'_1,s'_2) = \sup_{\mathbf{O} = (e_i)_i \in \mathcal{O}} D_c(e_i(s'_1),e_i(s'_2)) \mbox{ for } s'_1,s'_2 \in \mathcal{S} \enspace,
\end{equation}
where $D_c(p_i,q_i)$ denotes the classical trace distance ($L_1$ distance or Kolmogorov distance) \cite{ref:NC} between probability distributions $p_i$ and $q_i$: 
\begin{displaymath}
D_c(p_i,q_i) = \frac{1}{2} \sum_i |p_i - q_i| \enspace,
\end{displaymath}
and $\mathcal{O} = \bigcup_{N \in \mathbb{N}} \mathcal{O}_N$ denotes the set of all discrete observables.
(Note that the argument below shows that the supremum in \eqref{eq:TD} is always attained by some observable and it can be chosen from two-valued observables.)
Since the classical trace distance is the maximal difference of probabilities between $p_i$ and $q_i$ among all events $S$, i.e., $D_c(p_i,q_i) = \max_S |p(S) - q(S)| = \max_S |\sum_{i \in S} p_i - \sum_{i \in S} q_i| $, it is considered as an operationally natural distance between probability distributions.
In order to distinguish states $s_1'$ and $s_2'$ in general probabilistic theories, what one can do best is to find the best observable $\mathbf{O} = (e_i)_i \in \mathcal{O}$ for catching the difference between $s_1'$ and $s_2'$ by comparing the probability distributions $e_i(s'_1)$ and $e_i(s'_2)$.
Thus we are lead to the definition \eqref{eq:TD} of the distance between states; namely, $D(s'_1,s'_2)$ has the same operational meaning as Kolmogorov distance that is optimal among all observables.
From now, we show that Gudder's intrinsic metric \eqref{eq:relation_prob_intrinsic_metric} is in fact the same as our trace distance \eqref{eq:TD}.

For the purpose, first we show that in our trace distance, it suffices to consider just two-valued observables $\mathbf{O} = (e_i)_i \in \mathcal{O}_2$, namely:
\begin{equation}
\label{eq:TDwith2}
D(s'_1,s'_2)= \sup_{\mathbf{O} = (e_i)_i \in \mathcal{O}_2} D_c(e_i(s'_1),e_i(s'_2)) \enspace.
\end{equation}
(Now the supremum is attained by some observable due to the compactness of $\mathcal{O}_2$ and the continuity of $D_c(e_i(s'_1),e_i(s'_2))$; see the proof of Theorem \ref{thm:optimal_observable_exist}.)
To prove \eqref{eq:TDwith2}, note that one can associate to any $\mathbf{O} = (e_i)_i \in \mathcal{O}$ a two-valued observable $(e'_+,e'_-) \in \mathcal{O}_2$ with $e'_+ = \sum_{i \in M_+} e_i$ and $e'_- = 1 - e_+$, where $M_+ = \{ i \mid e_i(s'_1) \geq e_i(s'_2) \}$.
By the definition, we have $D_c(e_i(s'_1),e_i(s'_2)) = D_c(e'_{\pm}(s'_1),e'_{\pm}(s'_2))$.
This implies that the right-hand side of \eqref{eq:TDwith2} is greater than or equal to the right-hand side of \eqref{eq:TD}, while the opposite inequality holds obviously (since $\mathcal{O}_2 \subset \mathcal{O}$).
Hence \eqref{eq:TDwith2} holds.
Note that this argument also provides another simple expression of our trace distance $D(s'_1,s'_2)$:
\begin{equation}
\label{eq:TDwith3}
D(s'_1,s'_2)= \sup_{e \in \mathcal{E}} \left[ e(s'_1) -e(s'_2) \right] \enspace,
\end{equation}
where the supremum is again attained by some effect due to the compactness of $\mathcal{E}$ (see the proof of Theorem \ref{thm:optimal_observable_exist}).

Now it is not difficult to see that Gudder's intrinsic metric \eqref{eq:relation_prob_intrinsic_metric} is indeed the same as our trace distance \eqref{eq:TD}: To see this, just observe that for $s'_1,s'_2 \in \mathcal{S}$ with a priori probabilities $p_1=p_2=1/2$, we have from \eqref{eq:success_prob} and \eqref{eq:optimal_success_prob}
\begin{displaymath}
P_{\mathrm{succ}}(s'_1,s'_2) = \frac{1}{2} (1 + \sup_{e\in \mathcal{E}} \left[ e(s'_1) -e(s'_2) \right]) \enspace.
\end{displaymath}
Substituting it into \eqref{eq:relation_prob_intrinsic_metric} and using \eqref{eq:TDwith3}, we obtain the desired relation:
\begin{equation}
\label{eq:GIMTD}
\widetilde{d}(s'_1,s'_2) = D(s'_1,s'_2) \enspace.
\end{equation}

The equivalence \eqref{eq:GIMTD} provides simple proofs for several properties of $\widetilde{d}(s'_1,s'_2)$ originally shown by Gudder \cite{Gud_book79}.
For instance, since the classical trace distance $D_c(p_i,q_i)$ is well known to be a metric, so is our trace distance $D(s'_1,s'_2)$ by the definition, therefore $\widetilde{d}(s'_1,s'_2)$ is indeed a metric as well (for positiveness of $D(s'_1,s'_2)$ with $s'_1 \neq s'_2$ we needed the fact that the state space is separated).
We also consider another important property, the \emph{monotonicity} of $\widetilde{d}(s'_1,s'_2)$:
\begin{theorem}
[Gudder \cite{Gud_book79}]
\label{thm:GMon}
For any state $s'_1,s'_2 \in \mathcal{S}$ and any affine map $F: \mathcal{S} \to \mathcal{S}$, we have
\begin{displaymath}
\widetilde{d}(F(s'_1),F(s'_2)) \leq \widetilde{d}(s'_1,s'_2) \enspace.
\end{displaymath}
\end{theorem}
Now this fact is an easy consequence of the equivalence \eqref{eq:GIMTD} and the fact that affine maps are closed under composition.
Namely, for any observable $(e_i)_i \in \mathcal{O}$, by putting $f_i = e_i \circ F:\mathcal{S} \to \left[0,1\right]$ we have
\begin{equation}
\label{eq:relation_for_monotonicity}
D_c(e_i(F(s'_1)),e_i(F(s'_2))) = D_c(f_i(s'_1),f_i(s'_2)) \enspace.
\end{equation}
Now $(f_i)_i$ is also an observable, therefore the supremum of the left-hand side of \eqref{eq:relation_for_monotonicity} over $(e_i)_i \in \mathcal{O}$ does not exceed the supremum of the right-hand side of \eqref{eq:relation_for_monotonicity} over \emph{all} observables $(f_i)_i$.
This implies the monotonicity of $D(s'_1,s'_2)$, hence of $\widetilde{d}(s'_1,s'_2)$.
(We note that the quantity in the right-hand side of \eqref{eq:TDwith3} was also investigated in \cite{ref:Dariano} in slightly different context; for instance, it was shown to be a metric, and the monotonicity was also proven there.)

Summarizing, we have shown that Gudder's intrinsic metric has two operational meanings; one is directly given through the classical trace distance \eqref{eq:GIMTD}; another is given by the optimal success probability to discriminate two states under a uniform distribution \eqref{eq:relation_prob_intrinsic_metric}.
\end{remark}
\begin{remark}
\label{rem:NoiseDisturbance}
As an application of Gudder's intrinsic metric, or the trace distance defined above, we have a simple (qualitative) version of information disturbance theorem in general probabilistic theories.
Before giving the theorem, we clarify the meaning of some terminology.
We say that a state $s$ is a \emph{pure state} if $s$ is an extremal point of the state space.
We say that two states are \emph{indistinguishable} if these are not distinguishable in the sense of Definition \ref{defn:distinguishable}.
Then the above-mentioned theorem is the following:
\begin{theorem}
\label{thm:ND}
In any general probabilistic theory, any attempt to distinguish two indistinguishable pure states causes a disturbance. 
\end{theorem}
This theorem is a generalization of the well-known corresponding theorem in quantum theory (see e.g., Proposition 12.18 in \cite{ref:NC}) to arbitrary general probabilistic theories.
It is known that a general probabilistic theory is non-classical if and only if there exist indistinguishable pure states \cite{BBLW07}.
Hence one can conclude that the information disturbance property inevitably holds for any non-classical general probabilistic theory, not only for quantum theory.

Before presenting the proof, notice that any dynamics on $\mathcal{S}$ should be described by an affine map $F: \mathcal{S} \to \mathcal{S}$ in order to preserve the probabilistic mixture, while the composition of state spaces $\mathcal{S}_1$ and $\mathcal{S}_2$ is given by a tensor product $\mathcal{S}_1 \otimes \mathcal{S}_2$ (see \cite{BBLW07} and references therein).
\begin{proof}
[Theorem \ref{thm:ND}]
Let $s_1,s_2 \in \mathcal{S}$ be two indistinguishable pure states (thus $1 > P_{\mathrm{succ}}(s_1,s_2)$). 
Let $s_i \otimes s_0$ ($i=1,2$) be the initial states on $\mathcal{S} \otimes \mathcal{S}'$, where $s_0 \in \mathcal{S}'$ is any fixed state to which the information of $s_1$ or $s_2$ is transferred.
Assume contrary that one can extract information with which one distinguishes $s_1$ and $s_2$ without causing any disturbance.
More precisely, we assume that there exists an information transfer machine described by an affine map $F:\mathcal{S} \otimes \mathcal{S}' \to \mathcal{S} \otimes \mathcal{S}'$ such that the reduced states of $F(s_i \otimes s_0)$ to the first system $\mathcal{S}$ remains to be $s_i$ (i.e., causing no disturbance) while the reduced states of $F(s_1 \otimes s_0)$ and $F(s_2 \otimes s_0)$ to the second system $\mathcal{S}'$ are distinct (i.e., enabling one to extract some information to distinguish $s_1$ and $s_2$).
Now it is easy to show that if a reduced state is in pure state, then the whole state should be a product state by showing that there exist no correlations between an arbitrary pair of observables (or effects).
Therefore, we have
\begin{displaymath}
F(s_1 \otimes s_0) = s_1 \otimes t_1 \enspace,\enspace F(s_2 \otimes s_0) = s_2 \otimes t_2 \enspace,
\end{displaymath}
with $t_1 \neq t_2 \in \mathcal{S}'$.
Using the machine $F$ $N$ times, one obtains an affine transformation $\widetilde{F}$ on $\mathcal{S} \otimes \mathcal{S}'^{\otimes N}$ such that
\begin{displaymath}
\widetilde{F}(s_i \otimes (s_0^{\otimes N})) = s_i \otimes t_i^{\otimes N} \enspace.
\end{displaymath}
Physically, this means that one obtains an arbitrary large number of ensembles for (distinct) state $t_1$ or $t_2$, and thereby can distinguish them with success probability arbitrarily close to $1$.
In other words, the optimal success probability to distinguish $\widetilde{F}(s_1 \otimes (s_0^{\otimes N}))$ and $\widetilde{F}(s_1 \otimes (s_2^{\otimes N}))$ can be exponentially close to $1$ with respect to $N$ (to see this formally, use Chernoff bound \cite{Fuchs} for instance).
On the other hand, we have
\begin{displaymath}
1 > P_{\mathrm{succ}}(s_1,s_2)
= P_{\mathrm{succ}}(s_1 \otimes s_0^{\otimes N}, s_2 \otimes s_0^{ \otimes N})
\geq P_{\mathrm{succ}}(\widetilde{F}(s_1 \otimes s_0^{\otimes N}),\widetilde{F}(s_2 \otimes s_0^{\otimes N}))
\end{displaymath}
for any $N$, where the last inequality follows from Theorem \ref{thm:GMon} and \eqref{eq:relation_prob_intrinsic_metric}.
This is a contradiction, since the last term converges to $1$ when $N \to \infty$ as mentioned above.
Hence the proof of Theorem \ref{thm:ND} is concluded.
\end{proof}
\end{remark}

\subsection{Proof of Theorem \ref{thm:HE_exist}}
\label{subsec:proof_main_theorem}

In this subsection, we give a proof of Theorem \ref{thm:HE_exist}, namely we prove that $t_1$ and $t_2$ in $\widetilde{\mathcal{S}}$ are distinguishable if $(t_1,t_2) \in \mathcal{C}$ (see Definition \ref{defn:distinguishable} for terminology).

First, we would like to reduce our argument to the special case $t_2 = - t_1$.
For the purpose, let $v_0 = (t_1 + t_2)/2 \in \widetilde{\mathcal{S}}$ and put $C = \widetilde{\mathcal{S}} - v_0$, that is also a convex subset of $\widetilde{V}$.
Moreover, put $\overline{t_1} = t_1 - v_0$ and $\overline{t_2} = t_2 - v_0$.
Then we have $\overline{t_1},\overline{t_2} \in C$ and $\overline{t_2} = - \overline{t_1}$.
Note that $\overline{t_1} \neq \overline{t_2}$ since $t_1 \neq t_2$.

The outline of our proof is the following.
First, note that the existence of an $\widetilde{e} \in \widetilde{\mathcal{E}}$ such that $\widetilde{e}(t_1) = 1$ and $\widetilde{e}(t_2) = 0$ (that is nothing but our goal) is obvious if $\widetilde{V}$ coincides with the $1$-dimensional linear subspace $W'$ spanned by $\overline{t_1}$ (hence by $\overline{t_2}$).
To construct such an $\widetilde{e}$ in more general case, we would like to extend a nonzero linear functional $f$ on $W'$ (note that $f$ is continuous on $W'$ and $f(C \cap W')$ is bounded in $\mathbb{R}$) to a continuous linear functional $\overline{f}$ on $\widetilde{V}$ such that $\overline{f}(C)$ is bounded in $\mathbb{R}$.
Then it will be shown that the restriction of an appropriate affine transformation $h = \alpha \overline{f} + \beta$ of $\overline{f}$ ($\alpha,\beta \in \mathbb{R}$) to $\widetilde{\mathcal{S}}$ is the desired virtual effect $\widetilde{e}$.
To construct such an extension $\overline{f}$ of $f$, first we use Theorem \ref{thm:Hahn-Banach} to obtain an extension $f'$ of $f$ to $W = \widetilde{V}$ (not yet necessarily continuous) such that $f'(C)$ is bounded in $\mathbb{R}$, and then we further modify the functional $f'$ by using Theorem \ref{thm:extension_continuous_to_continuous} to obtain $\overline{f}$.

To perform the program, we start with the linear functional $f$ on the $1$-dimensional subspace $W'$ such that $f(\lambda \overline{t_1}) = \lambda$ for each $\lambda \in \mathbb{R}$, therefore $f(\overline{t_1}) = 1$ and $f(\overline{t_2}) = -1$.
To apply Theorem \ref{thm:Hahn-Banach}, we would like to take an appropriate semi-norm $g$ on $\widetilde{V}$, more precisely, the Minkowski functional $g_{\widetilde{C}}$ of a certain subset $\widetilde{C}$ of $\widetilde{V}$ (see Proposition \ref{prop:Minkowski_functional}).
From now, we define the subset $\widetilde{C}$.
Note that the convex subset $C$ of $\widetilde{V}$ contains the origin of $\widetilde{V}$, therefore we have $\lambda x \in C$ for any $x \in C$ and $0 \leq \lambda \leq 1$.
Thus the subset $\pm C = C \cup -C$ of $\widetilde{V}$ is circled (see Sect.\ \ref{subsec:HF_auxiliary_result} for terminology).
Now define $\widetilde{C}$ to be the convex hull $\mathrm{Conv}(\pm C)$ of $\pm C$, which is also a circled subset of $\widetilde{V}$.
By the convexity of $C$, any element $v$ of $\widetilde{C}$ can be written as $v = \lambda x - \lambda' x'$ with $x,x' \in C$, $\lambda,\lambda' \geq 0$ and $\lambda + \lambda' = 1$.
This subset $\widetilde{C}$ has the following property:
\begin{lemma}
\label{lem:subset_radial}
$\widetilde{C}$ is a radial subset of $\widetilde{V}$ (see Sect.\ \ref{subsec:HF_auxiliary_result} for terminology).
\end{lemma}
\begin{proof}
Let $W_0$ be the set of all $v \in \widetilde{V}$ such that $v \in \lambda \widetilde{C}$ for some $\lambda > 0$.
Then $W_0$ contains $\widetilde{C}$, hence $C$.
Moreover, if $v \in W_0$, $\lambda_0 > 0$ and $v \in \lambda_0 \widetilde{C}$, then we have $v \in \lambda \widetilde{C}$ whenever $|\lambda| \geq \lambda_0$ since $\widetilde{C}$ is circled.
Thus $\widetilde{C}$ is radial if $W_0 = \widetilde{V}$.
To prove $W_0 = \widetilde{V}$, it suffices to show that $W_0$ is a linear subspace of $\widetilde{V}$.
Indeed, once this is proven, $W_0 + v_0$ will be an affine subspace of $\widetilde{V}$ containing $\widetilde{\mathcal{S}}$ (recall that $W_0 \supset \widetilde{C} = \widetilde{\mathcal{S}} - v_0$), therefore $W_0 + v_0 = \widetilde{V}$ (hence $W_0 = \widetilde{V}$) since $\mathrm{Aff}(\widetilde{\mathcal{S}}) = \widetilde{V}$.

Let $v_1,v_2 \in W_0$.
Then for each $i$, we have $v_i \in \lambda_i x_i$ for some $\lambda_i > 0$ and $x_i \in \widetilde{C}$.
Moreover, let $\mu_1,\mu_2 \in \mathbb{R}$, and write $\mu_i = \varepsilon_i \nu_i$ with $\varepsilon_i \in \{\pm 1\}$ and $\nu_i \geq 0$ for each $i$.
We show that $\mu_1 v_1 + \mu_2 v_2 \in W_0$; since this is obvious when $\mu_1 = \mu_2 = 0$, we assume from now that $\nu_1 > 0$ or $\nu_2 > 0$.
Then by putting $x'_i = \varepsilon_i x_i \in \widetilde{C}$ for each $i$ (note that $\widetilde{C}$ is circled), we have
\begin{displaymath}
\mu_1 v_1 + \mu_2 v_2
= \lambda_1 \nu_1 x'_1 + \lambda_2 \nu_2 x'_2
= (\lambda_1 \nu_1 + \lambda_2 \nu_2) \frac{ \lambda_1 \nu_1 x'_1 + \lambda_2 \nu_2 x'_2 }{ \lambda_1 \nu_1 + \lambda_2 \nu_2 } \enspace,
\end{displaymath}
therefore $\mu_1 v_1 + \mu_2 v_2 \in (\lambda_1 \nu_1 + \lambda_2 \nu_2) \widetilde{C}$ by the convexity of $\widetilde{C}$.
Hence we have $\mu_1 v_1 + \mu_2 v_2 \in W_0$, therefore Lemma \ref{lem:subset_radial} holds.
\end{proof}
Owing to the above properties of $\widetilde{C}$, we define the semi-norm $g$ to be the Minkowski functional $g_{\widetilde{C}}$ of $\widetilde{C}$ (see Proposition \ref{prop:Minkowski_functional}).
Note that $g(v) \leq 1$ for any $v \in \widetilde{C}$ by the definition of $g = g_{\widetilde{C}}$.

From now, to apply Theorem \ref{thm:Hahn-Banach}, we show that $|f(v)| \leq g(v)$ for any $v \in W'$.
Since $W'$ is $1$-dimensional and $g$ is a semi-norm, it suffices to show that $g(\overline{t_1}) = 1 = f(\overline{t_1})$.
This is proven in the following lemma:
\begin{lemma}
\label{lem:t_i_extremal}
We have $g(\overline{t_1}) = 1$.
\end{lemma}
\begin{proof}
First, note that $g(\overline{t_1}) \leq 1$ since $\overline{t_1} \in \widetilde{C}$.
We show that $g(\overline{t_1}) \geq 1$, or equivalently, there does not exist an element $v \in \widetilde{C}$ and $c > 1$ such that $v = c \overline{t_1}$.
Assume contrary that such a pair $(v,c)$ exists.
As mentioned before, this $v$ is of the form $v = \lambda x - (1 - \lambda) x'$ with $x,x' \in C$ and $0 \leq \lambda \leq 1$, therefore $\lambda x - (1 - \lambda) x' = c \overline{t_1} = -c \overline{t_2}$.
Moreover, by the definition of $C$, we have $x = s - v_0$ and $x' = s' - v_0$ for some $s,s' \in \widetilde{\mathcal{S}}$, therefore
\begin{displaymath}
v = \lambda s - (1 - \lambda) s' + (1 - 2 \lambda) v_0 = c \overline{t_1} = -c \overline{t_2} \enspace.
\end{displaymath}
Note also that $t_2 - t_1 = \overline{t_2} - \overline{t_1} = 2 \overline{t_2} = -2 \overline{t_1}$.
From now, we show that we can construct a pair $(t'_1,t'_2) \in \mathcal{C}_{\mathrm{weak}}'$ (by using the convexity of $\widetilde{\mathcal{S}}$) such that $\ell(t'_1,t'_2) > \ell(t_1,t_2)$, contradicting the assumption $(t_1,t_2) \in \mathcal{C}$.

First we consider the case that $p_1 = p_2$.
If $\lambda \leq 1/2$, then we have
\begin{displaymath}
(1 - \lambda) s' - \lambda s - (1 - 2\lambda) t_1 = -v - (1 - 2 \lambda) \overline{t_1} = (c + 1 - 2\lambda) \overline{t_2} \enspace,
\end{displaymath}
therefore $s' - s'' = \alpha (t_2 - t_1)$, where $s'' = (\lambda s + (1 - 2\lambda) t_1) / (1 - \lambda) \in \widetilde{\mathcal{S}}$ (note that $\widetilde{\mathcal{S}}$ is convex) and $\alpha = (c + 1 - 2\lambda) / (2 - 2\lambda)$.
Since $c > 1$, we have $\alpha > 1$, therefore $(s'',s') \in \mathcal{C}'_{\mathrm{weak}}$ and $\ell(s'',s') = \alpha \ell(t_1,t_2) > \ell(t_1,t_2)$, as desired.
Similarly, if $\lambda \geq 1/2$, then we have
\begin{displaymath}
(2\lambda - 1) t_2 + (1 - \lambda) s' - \lambda s = -v + (2\lambda - 1) \overline{t_2} = (c + 2\lambda - 1) \overline{t_2} \enspace,
\end{displaymath}
therefore $s'' - s = \alpha (t_2 - t_1)$ where $s'' = (2 - \lambda^{-1}) t_2 + (\lambda^{-1} - 1) s' \in \widetilde{\mathcal{S}}$ and $\alpha = (c + 2\lambda - 1) / (2\lambda) > 1$.
Thus we have $(s,s'') \in \mathcal{C}'_{\mathrm{weak}}$ and $\ell(s,s'') > \ell(t_1,t_2)$, as desired.

Secondly, we consider the case that $p_1 > p_2$ and $s^{\ast} \not\in \widetilde{\mathcal{S}}$.
Put $\ell = \ell(t_1,t_2)$ for simplicity.
Note that $0 < \ell < 1$ and
\begin{displaymath}
\ell s^{\ast} = t_2 - (1 - \ell) t_1 = 2 v_0 - (2 - \ell) t_1 = (2 - \ell) t_2 - (2 - 2 \ell) v_0 \enspace,
\end{displaymath}
while
\begin{displaymath}
v = \lambda s - (1 - \lambda) s' + (1 - 2 \lambda) v_0 = c t_1 - c v_0 = c v_0 - c t_2 \enspace.
\end{displaymath}
Put $\mu = (2 - \ell)(2\lambda - 1) + c \ell$.
If $\mu \geq 0$, then the above relations imply that
\begin{displaymath}
\lambda (2 - 2\ell) s + \ell(c + 2 \lambda - 1) s^{\ast}
= (1 - \lambda)(2 - 2\ell) s' + \mu t_2 \enspace.
\end{displaymath}
Now the coefficients of $s$, $s^{\ast}$, $s'$, and $t_2$ in this equality are all nonnegative, and the sums of the two coefficients in the left-hand side and in the right-hand side, respectively, are positive and equal to each other; namely,
\begin{displaymath}
\lambda (2 - 2\ell) + \ell(c + 2\lambda - 1) = (1 - \lambda)(2 - 2\ell) + \mu = c\ell + 2 \lambda - \ell > 0 \enspace.
\end{displaymath}
Thus by the convexity of $\widetilde{\mathcal{S}}$, we have $(1 - \alpha) s + \alpha s^{\ast} = s''$ for some $s'' \in \widetilde{\mathcal{S}}$, where
\begin{displaymath}
\alpha = \frac{ \ell( c + 2\lambda - 1) }{ c\ell + 2\lambda - \ell }
= 1 - \frac{ 2 \lambda (1 - \ell) }{ c\ell + 2\lambda - \ell } \in \left(\ell,1\right]
\end{displaymath}
(note that $0 < \ell < 1$ and $c > 1$).
Thus we have $(s,s'') \in \mathcal{C}'_{\mathrm{weak}}$ and $\ell(s,s'') = \alpha > \ell$, as desired.
Similarly, if $\mu < 0$, then we have
\begin{displaymath}
2 \lambda s + |\mu| t_1 + \ell(c + 1 - 2 \lambda) s^{\ast} = (2 - 2\lambda) s' \enspace.
\end{displaymath}
Since $c > 1$, all the four coefficients in this equality are nonnegative, and the sum of the three coefficients in the left-hand side is equal to the coefficient $2 - 2\lambda > 0$ in the right-hand side; namely, 
\begin{displaymath}
2 \lambda + |\mu| + \ell(c + 1 - 2 \lambda) = 2 - 2\lambda > 0 \enspace.
\end{displaymath}
Thus by the convexity of $\widetilde{\mathcal{S}}$, we have $(1 - \alpha) s'' + \alpha s^{\ast} = s'$ for some $s'' \in \widetilde{\mathcal{S}}$, where $\alpha = \ell(c + 1 - 2\lambda) / (2 - 2\lambda) \in \left(\ell,1\right]$ (note that $\ell > 0$ and $c > 1$).
Thus we have $(s'',s') \in \mathcal{C}'_{\mathrm{weak}}$ and $\ell(s'',s') = \alpha > \ell$, as desired.

Hence our claim holds in all cases, therefore Lemma \ref{lem:t_i_extremal} holds.
\end{proof}
Thus by Theorem \ref{thm:Hahn-Banach}, the functional $f$ on $W'$ extends to an $f' \in \mathcal{L}(\widetilde{V})$ such that $|f'(v)| \leq g(v)$ for any $v \in \widetilde{V}$.
Since $f'|_{W'} = f$, we have $f'(\overline{t_1}) = 1$, $f'(\overline{t_2}) = -1$ and $|f'(x)| \leq g(x) \leq 1$ for any $x \in \widetilde{C}$, therefore $f'(C) \subset \left[-1,1\right]$.
By putting $\alpha = f'(v_0)$, it follows that
\begin{displaymath}
f'(t_1) = \alpha + 1 \enspace,\enspace f'(t_2) = \alpha - 1 \enspace,\enspace f'(\widetilde{\mathcal{S}}) \subset \left[\alpha - 1,\alpha + 1\right] \enspace,
\end{displaymath}
therefore the restriction of $f'$ to $V$ is continuous.
Our desired virtual effect $\widetilde{e}$ can be constructed directly from this $f'$ if $f'$ is also continuous on $\widetilde{V}$; however, this is not guaranteed in general.

Thus, instead, by using Theorem \ref{thm:extension_continuous_to_continuous}, we take a continuous linear functional $\overline{f}$ on $\widetilde{V}$ such that $\overline{f}|_{V} = f'|_{V}$.
Note that $\overline{f}(\mathcal{S}) \subset \left[\alpha - 1,\alpha + 1\right]$ since $\mathcal{S} \subset \widetilde{\mathcal{S}} \cap V$, therefore we have $\overline{f}(\widetilde{\mathcal{S}}) \subset \left[\alpha - 1,\alpha + 1\right]$ since $\widetilde{\mathcal{S}} = \mathrm{cl}_{\widetilde{V}}(\mathcal{S})$.
From now, we show that $\overline{f}(t_1) = \alpha + 1$ and $\overline{f}(t_2) = \alpha - 1$.
First, we consider the case $p_1 = p_2$.
Then we have $t_1 - t_2 = c(s_2 - s_1)$ with $c = \ell(t_1,t_2) > 0$, while $s_2 - s_1 \in V$ since $s_1,s_2 \in \mathcal{S}$, therefore
\begin{displaymath}
\overline{f}(t_2 - t_1) = c \overline{f}(s_1 - s_2) = c f'(s_1 - s_2) = f'(t_2 - t_1) = -2 \enspace.
\end{displaymath}
Since $\overline{f}(\widetilde{\mathcal{S}}) \subset \left[\alpha - 1,\alpha + 1\right]$ as mentioned above, we have
\begin{displaymath}
\alpha - 1 \leq \overline{f}(t_2) = \overline{f}(t_1) - 2 \leq \alpha + 1 - 2 = \alpha - 1 \enspace,
\end{displaymath}
therefore $\overline{f}(t_2) = \alpha - 1$ and $\overline{f}(t_1) = \overline{f}(t_2) + 2 = \alpha + 1$.
Secondly, we consider the case $p_1 > p_2$ and $s^{\ast} \not\in \widetilde{\mathcal{S}}$.
Now $t_2 = c s^{\ast} + (1 - c) t_1$ with $0 < c = \ell(t_1,t_2) < 1$, while $s^{\ast} \in V$, therefore
\begin{displaymath}
\overline{f}(t_2) - (1-c) \overline{f}(t_1) = c \overline{f}(s^{\ast}) = c f'(s^{\ast}) = f'(t_2) - (1 - c) f'(t_1) \enspace.
\end{displaymath}
Now we have $\overline{f}(t_1) \leq \alpha + 1 = f'(t_1)$, therefore $\overline{f}(t_2) \leq f'(t_2) = \alpha - 1$ since $1 - c > 0$.
Thus we have $\overline{f}(t_2) = \alpha - 1$ since $\overline{f}(t_2) \geq \alpha - 1$, therefore $\overline{f}(t_1) = f'(t_1) = \alpha + 1$.
Hence we have $\overline{f}(t_1) = \alpha + 1$ and $\overline{f}(t_2) = \alpha - 1$ in any case.

Finally, by the above properties, the affine functional $h = (\overline{f} + 1 - \alpha) / 2$ on $\widetilde{V}$ is continuous and satisfies that $h(t_1) = 1$, $h(t_2) = 0$ and $h(\widetilde{\mathcal{S}}) \subset \left[0,1\right]$.
This implies that $\widetilde{e} = h|_{\widetilde{\mathcal{S}}}$ is a virtual effect that distinguishes $t_1$ and $t_2$.

Hence the proof of Theorem \ref{thm:HE_exist} is concluded.

\paragraph*{Acknowledgments.}
The authors would like to thank Dr.\ Manabu Hagiwara, Dr.\ Kentaro Imafuku, and Professor Hideki Imai, for their significant comments.
A part of this work was supported by Grant-in-Aid for Young Scientists (B) (20700017), The Ministry of Education, Culture, Sports, Science and Technology (MEXT).

\appendix
\section*{Appendix: Proof of Theorem \ref{thm:state_space_embedded}}
In the appendix, we give a proof of Theorem \ref{thm:state_space_embedded}.
In what follows, 
For any convex structure $C$, let $\mathcal{A}^b_{C'}(C)$ be the set of all $f \in \mathcal{A}(C)$ bounded on a subset $C'$ of $C$.
Moreover, for any convex subset $C$ of a t.v.s., let $\mathcal{A}_c(C)$ denote the set of all continuous $f \in \mathcal{A}(C)$.

\section{Construction of $\mathcal{S}$ and $V$}
\label{sec:appendix_S_V}
First, we describe construction of a vector space $V$ and its convex subset $\mathcal{S}$ such that $\mathcal{S}$ is isomorphic to the separated convex structure $\overline{\mathcal{S}_0}$ and $V = \mathrm{Aff}(\mathcal{S})$.
Here we abuse the notations $\mathcal{S}$ and $V$ though these $\mathcal{S}$ and $V$ are in fact not necessarily the same as (but isomorphic to) $\mathcal{S}$ and $V$ in Theorem \ref{thm:state_space_embedded}, respectively.
Although our argument is essentially the standard one (cf., \cite{BBLW07,Gud_book79,Hol_book,KMI08,Mac_book,Oza80}), we give the argument here for the sake of completeness.

Our argument is the following.
In what follows, let $\mathcal{L}(W)$ denote the set of all linear functionals on a vector space $W$; and for any convex structure $C$, let $\mathcal{A}(C)$ denote the set of all affine functionals on $C$.
Then the set $\mathcal{A}(\overline{\mathcal{S}_0})$ forms a vector space with natural addition and scalar multiplication, therefore its dual space $\mathcal{A}(\overline{\mathcal{S}_0})^{\ast} = \mathcal{L}(\mathcal{A}(\overline{\mathcal{S}_0}))$ is also a vector space.
We define an \lq\lq evaluation map'' $\mathsf{ev}_s:\mathcal{A}(\overline{\mathcal{S}_0}) \to \mathbb{R}$ for each $s \in \overline{\mathcal{S}_0}$ by $\mathsf{ev}_s(f) = f(s)$ for $f \in \mathcal{A}(\overline{\mathcal{S}_0})$.
Then a straightforward argument shows that $\mathsf{ev}_s \in \mathcal{A}(\overline{\mathcal{S}_0})^{\ast}$ for every $s \in \overline{\mathcal{S}_0}$, and the map $\psi:\overline{\mathcal{S}_0} \to \mathcal{A}(\overline{\mathcal{S}_0})^{\ast}$, $\psi(s) = \mathsf{ev}_s$, is a homomorphism of convex structures, i.e., $\psi(\langle \lambda,\mu;s,t \rangle) = \lambda \psi(s) + \mu \psi(t)$ for any $s,t \in \overline{\mathcal{S}_0}$.
The fact that $\overline{\mathcal{S}_0}$ is separated (Lemma \ref{lem:separated}) implies that $\psi$ is injective.
Moreover, by fixing an element $v \in \psi(\overline{\mathcal{S}_0})$, the map $\varphi:\overline{\mathcal{S}_0} \to \mathcal{A}(\overline{\mathcal{S}_0})^{\ast}$, $\varphi(s) = \psi(s) - v$, is also an injective homomorphism of convex structures.
Thus $\mathcal{S} = \varphi(\overline{\mathcal{S}_0})$ is a convex subset of the vector space $\mathcal{A}(\overline{\mathcal{S}_0})^{\ast}$ containing the origin of $\mathcal{A}(\overline{\mathcal{S}_0})^{\ast}$.
Now $V = \mathrm{Aff}(\mathcal{S})$ is a linear subspace of $\mathcal{A}(\overline{\mathcal{S}_0})^{\ast}$.
Thus $\mathcal{S}$ and $V$ are obtained.

\section{Topologies on $\mathcal{S}$ and $V$}
\label{sec:appendix_topology_S_V}
Secondly, we give the definition of topologies on $V$ and $\mathcal{S}$.
In what follows, for any vector space $W$, let $\mathcal{L}^b_{C}(W)$ denote the set of all $f \in \mathcal{L}(W)$ bounded on a given subset $C$ of $W$.
For any t.v.s.\ $W$, let $\mathcal{L}_c(W)$ denote the set of all continuous $f \in \mathcal{L}(W)$.
For a convex subset $C$ of a vector space $W$ and a subset $\mathcal{F}$ of $\mathcal{A}(C)$, let $\sigma(C,\mathcal{F})$ denote the weakest topology on $C$ to make every $f \in \mathcal{F}$ continuous.
For a topology $\mathcal{T}$ on a space $X$ and a subset $Y$ of $X$, let $\mathcal{T}|_Y$ denote the relative topology on $Y$ induced by $\mathcal{T}$.
For two topologies $\mathcal{T}$ and $\mathcal{T}'$ on the same set $X$, we write $\mathcal{T} \subset \mathcal{T}'$ to signify that $\mathcal{T}'$ is stronger than or equal to $\mathcal{T}$ (i.e., every $\mathcal{T}$-open subset of $X$ is $\mathcal{T}'$-open).
Moreover, let $\mathcal{E}$ denote the set of all $e \in \mathcal{A}(S)$ such that $e(\mathcal{S}) \subset \left[0,1\right]$.

Now we define the topology $\mathcal{T}(V)$ on $V$ by
\begin{displaymath}
\mathcal{T}(V) = \sigma(V,\mathcal{L}^b_{\mathcal{S}}(V)) \enspace.
\end{displaymath}
This topology makes $V$ a l.c.t.v.s.\ (see e.g., \cite[Chap.\ II, Sect.\ 5]{SW_book}).
Moreover, this $V$ satisfies the following property:
\begin{lemma}
\label{lem:Hausdorff_V}
The t.v.s.\ $V$ is Hausdorff.
\end{lemma}
\begin{proof}
First, since $\mathcal{S}$ is convex, an elementary argument shows that the affine hull $\mathrm{Aff}(\mathcal{S}) = V$ of $\mathcal{S}$ consists of all elements of the form $\lambda s - \lambda' s'$ with $s,s' \in \mathcal{S}$, $\lambda \geq 1$ and $\lambda - \lambda' = 1$.
Let $v = \lambda s - \lambda' s'$ and $v' = \mu t - \mu' t'$ be distinct elements of $V$ written in the above form.
Now put
\begin{displaymath}
p = \frac{ \lambda - 1 }{ \lambda + \mu - 1} \enspace,\enspace
q = \frac{ \mu - 1 }{ \lambda + \mu - 1 } \enspace,\enspace
r = \frac{ 1 }{ \lambda + \mu - 1 } \enspace,
\end{displaymath}
therefore $p,q \geq 0$, $r > 0$ and $p + q + r = 1$.
Moreover, put
\begin{displaymath}
w = p s' + q t' + r v \enspace,\enspace
w' = p s' + q t' + r v' \enspace.
\end{displaymath}
Then $w \neq w'$ since $v \neq v'$ and $r > 0$, while we have
\begin{displaymath}
w = r \lambda s + (p - r \lambda') s' + q t' = (1 - q) s + q t' \in \mathcal{S}
\end{displaymath}
since $\mathcal{S}$ is convex, and similarly $w' \in \mathcal{S}$.
Since $\mathcal{S} \simeq \overline{\mathcal{S}_0}$ is separated by Lemma \ref{lem:separated}, there exists an $e \in \mathcal{E}$ such that $e(w) \neq e(w')$.
Now by the definitions of $w$ and $w'$, the affine extension $f$ of $e$ to $V$ satisfies $f \in \mathcal{L}^b_{\mathcal{S}}(V)$ and $f(v) \neq f(v')$.
Thus $V$ is Hausdorff with respect to $\sigma(V,\mathcal{L}^b_{\mathcal{S}}(V))$.
Hence Lemma \ref{lem:Hausdorff_V} holds.
\end{proof}

On the other hand, the induced topology on $\mathcal{S}$ satisfies the following:
\begin{lemma}
\label{lem:induced_topology_S}
Two topologies $\mathcal{T}(V)|_{\mathcal{S}}$ and $\sigma(\mathcal{S},\mathcal{E})$ on $\mathcal{S}$ coincide.
\end{lemma}
\begin{proof}
In the proof, put $\mathcal{T} = \mathcal{T}(V) = \sigma(V,\mathcal{L}^b_{\mathcal{S}}(V))$.
First, we show that each $e \in \mathcal{E}$ is ($\mathcal{T}|_{\mathcal{S}}$)-continuous.
Since $\mathrm{Aff}(\mathcal{S}) = V$, this $e$ extends to an affine functional $f$ on $V$ such that $f(\mathcal{S})$ is bounded, therefore $f + \alpha \in \mathcal{L}^b_{\mathcal{S}}(V)$ for some $\alpha \in \mathbb{R}$.
Thus $f + \alpha$ is $\mathcal{T}$-continuous by the definition of $\mathcal{T}$, therefore $f$ is also $\mathcal{T}$-continuous and $e = f|_{\mathcal{S}}$ is ($\mathcal{T}|_{\mathcal{S}}$)-continuous as desired.
This implies that $\sigma(\mathcal{S},\mathcal{E}) \subset \mathcal{T}|_{\mathcal{S}}$.

Now it suffices to show that each ($\mathcal{T}|_{\mathcal{S}}$)-open subset $U$ of $\mathcal{S}$ is $\sigma(\mathcal{S},\mathcal{E})$-open.
Take a $\mathcal{T}$-open subset $U'$ of $V$ such that $U = U' \cap \mathcal{S}$.
Then for each $s \in U \subset U'$, by the definition of $\mathcal{T}$, there exist a finite number of $f_i \in \mathcal{L}^b_{\mathcal{S}}(V)$ and the same number of open subsets $W_i \subset \mathbb{R}$ such that $s \in \bigcap_i f_i^{-1}(W_i) \subset U'$.
Since $s \in \mathcal{S}$, we have $s \in \bigcap_i (\mathcal{S} \cap f_i^{-1}(W_i)) \subset U$, therefore it suffices to show that each subset $\mathcal{S} \cap f_i^{-1}(W_i)$ of $\mathcal{S}$ is $\sigma(\mathcal{S},\mathcal{E})$-open.
Since $f_i(\mathcal{S})$ is bounded, there exist $\alpha_i,\beta_i \in \mathbb{R}$ such that $\alpha_i \neq 0$ and the functional $g_i = \alpha_i f_i + \beta_i$ satisfies $g_i(\mathcal{S}) \subset \left[0,1\right]$, therefore $e_i = g_i|_{\mathcal{S}} \in \mathcal{E}$.
Moreover, we have $f_i^{-1}(W_i) = g_i^{-1}(\alpha_i W_i + \beta_i)$ and $W'_i = \alpha_i W_i + \beta_i$ is also an open subset of $\mathbb{R}$.
Thus $\mathcal{S} \cap f_i^{-1}(W_i) = \mathcal{S} \cap g_i^{-1}(W'_i) = e_i^{-1}(W'_i)$, that is $\sigma(\mathcal{S},\mathcal{E})$-open by the definition of $\sigma(\mathcal{S},\mathcal{E})$.
Hence Lemma \ref{lem:induced_topology_S} holds.
\end{proof}

\section{The Completions of $\mathcal{S}$ and $V$}
\label{sec:appendix_completion_S_V}
To proceed the proof of Theorem \ref{thm:state_space_embedded} further, we recall the following notion: The \emph{completion} of a uniform space $X$ is a complete uniform space $\widetilde{X}$ such that $X$ is a dense subspace of $\widetilde{X}$.
(See e.g., \cite[Chap.\ II]{Bou_book} or \cite{SW_book} for properties of uniform spaces).
The completion $\widetilde{X}$ of such a space $X$ always exists, and $\widetilde{X}$ is Hausdorff if and only if $X$ is Hausdorff.
Since any t.v.s.\ is a uniform space (see e.g., Proposition 1.4 in \cite[Chap.\ I]{SW_book}), the completion $\widetilde{V}$ of the Hausdorff t.v.s.\ $V$ exists in the above sense.
Moreover, this $\widetilde{V}$ also admits a structure of a t.v.s., and now $\widetilde{V}$ is a complete Hausdorff t.v.s.\ and $V$ is a topological vector subspace of $\widetilde{V}$ (with the induced topology equal to $\sigma(V,\mathcal{L}^b_{\mathcal{S}}(V))$) that is dense in $\widetilde{V}$ (see e.g., Proposition 1.5 in \cite[Chap.\ I]{SW_book}).
Here we use the conventional notation $\widetilde{V}$ for the completion of $V$, though it is not necessarily the same as (but is closely related to) the $\widetilde{V}$ in Theorem \ref{thm:state_space_embedded}.

Since $\mathcal{S}$ is convex, the closure $\widetilde{\mathcal{S}} = \mathrm{cl}_{\widetilde{V}}(\mathcal{S})$ of $\mathcal{S}$ in $\widetilde{V}$ is also convex in $\widetilde{V}$ (see e.g., Proposition 1.2 in \cite[Chap.\ II]{SW_book}).
Again, note that this $\widetilde{\mathcal{S}}$ does not necessarily coincide with (but is closely related to) the $\widetilde{\mathcal{S}}$ in Theorem \ref{thm:state_space_embedded}.
Now the closed subset $\widetilde{\mathcal{S}}$ of the complete t.v.s.\ $\widetilde{V}$ is also complete (as a uniform subspace), therefore $\widetilde{\mathcal{S}}$ is the completion of $\mathcal{S}$ (as a uniform subspace of $V$) since $\mathcal{S}$ is dense in $\widetilde{\mathcal{S}}$.
We would like to show that $\widetilde{\mathcal{S}}$ is compact; we give a lemma for the purpose.
Here we use the following terminology.
A subset $B$ of a t.v.s.\ $W$ is called \emph{bounded} if for any $0$-neighborhood (i.e., neighborhood of the origin) $U$ of $W$, there exists a $\lambda \in \mathbb{R}$ such that $B \subset \lambda U$.
Then we have the following:
\begin{lemma}
\label{lem:state_space_bounded}
The convex subset $\mathcal{S}$ of $V$ is bounded in $V$.
\end{lemma}
\begin{proof}
By the definition of the topology on $V$, each $0$-neighborhood $U$ of $V$ contains an open $0$-neighborhood of the form $\bigcap_i f_i^{-1}(U'_i)$ with finitely many $f_i \in \mathcal{L}^b_{\mathcal{S}}(V)$ and the same number of open subsets $U'_i$ of $\mathbb{R}$ containing $0$.
Since each $f_i(\mathcal{S}) \subset \mathbb{R}$ is bounded, there is a $\lambda > 0$ such that $f_i(\mathcal{S}) \subset \lambda U'_i$ for every $i$.
Thus $\mathcal{S} \subset \lambda f_i^{-1}(U'_i)$ for every $i$, therefore $\mathcal{S} \subset \lambda U$.
Hence the lemma holds.
\end{proof}
Now note that the topology $\mathcal{T}(V) = \sigma(V,\mathcal{L}^b_{\mathcal{S}}(V))$ of $V$ is a weak topology, i.e., it coincides with $\sigma(V,\mathcal{L}_c(V))$ where continuity of each $f \in \mathcal{L}_c(V)$ is with respect to $\mathcal{T}(V)$ (namely, every member of $\mathcal{L}^b_{\mathcal{S}}(V)$ is continuous with respect to $\sigma(V,\mathcal{L}_c(V))$ and every member of $\mathcal{L}_c(V)$ is continuous with respect to $\mathcal{T}(V)$).
Since $\mathcal{S} \subset V$ is bounded by Lemma \ref{lem:state_space_bounded}, and $V$ is l.c., it follows that $\mathcal{S}$ is \emph{precompact}, i.e., the completion $\widetilde{\mathcal{S}}$ of $\mathcal{S}$ is compact (see e.g., Corollary 2 of Proposition 5.5 in \cite[Chapter IV]{SW_book}).
The current situation is summarized as follows:
\begin{itemize}
\item $\mathcal{S} \simeq \overline{\mathcal{S}_0}$ is a convex subset of a l.c.\ Hausdorff t.v.s.\ $V$ containing the origin, with $\mathrm{Aff}(\mathcal{S}) = V$, such that the induced topology on $\mathcal{S}$ is $\sigma(\mathcal{S},\mathcal{E})$;
\item the topology $\mathcal{T}(V)$ of $V$ is $\sigma(V,\mathcal{L}^b_{\mathcal{S}}(V)) = \sigma(V,\mathcal{L}_c(V))$;
\item $\widetilde{V}$ is a complete Hausdorff t.v.s.\ containing $V$ as a dense topological vector subspace;
\item $\widetilde{\mathcal{S}} = \mathrm{cl}_{\widetilde{V}}(\mathcal{S})$ is the completion of $\mathcal{S}$ that is compact and convex.
\end{itemize}

\section{Existence of the Objects in Theorem \ref{thm:state_space_embedded}}
\label{sec:appendix_modification_S_V}
From now, we modify the above objects to obtain the objects in Theorem \ref{thm:state_space_embedded}.
In what follows, for a t.v.s.\ $W$, let $\sigma(W)$ denote the weak topology $\sigma(W,\mathcal{L}_c(W))$ on $W$.
The following facts will be used in our argument:
\begin{proposition}
[{Corollary 2 of Theorem 4.1 in \cite[Chap.\ IV]{SW_book}}]
\label{prop:weak_top_quotient}
Let $W$ be a l.c.t.v.s.\ with topology $\mathcal{T}$, $W'$ a vector subspace of $W$, and $\overline{W} = W/W'$ the quotient space.
Then the weak topology $\sigma(W')$ on $W'$ with respect to $\mathcal{T}|_{W'}$ coincides with $\sigma(W)|_{W'}$, and the weak topology $\sigma(\overline{W})$ on $\overline{W}$ with respect to the quotient topology induced by $\mathcal{T}$ is the quotient topology induced by $\sigma(W)$.
\end{proposition}
\begin{theorem}
[{Theorem 4.2 in \cite[Chap.\ II]{SW_book}}]
\label{thm:extension_continuous_to_continuous}
Let $W$ be a l.c.t.v.s., $W'$ a vector subspace of $W$, and $f \in \mathcal{L}_c(W')$.
Then $f$ extends to an $\overline{f} \in \mathcal{L}_c(W)$.
\end{theorem}

Note that the weak topology $\sigma(\widetilde{V})$ on $\widetilde{V}$ with respect to the original topology $\mathcal{T}$ of $\widetilde{V}$ is weaker than or equal to $\mathcal{T}$, therefore $\widetilde{\mathcal{S}}$ is also compact with respect to $\sigma(\widetilde{V})$.
Now we have the following property:
\begin{lemma}
\label{lem:induced_on_S_coincide}
We have $\sigma(\widetilde{V})|_V = \mathcal{T}(V) = \sigma(V,\mathcal{L}^b_{\mathcal{S}}(V))$.
\end{lemma}
\begin{proof}
Note that $\sigma(\widetilde{V})|_V \subset \sigma(V,\mathcal{L}^b_{\mathcal{S}}(V))$ since $\mathcal{T}|_V = \mathcal{T}(V)$ by the definition of $\widetilde{V}$.
Thus it suffices to show that each $f \in \mathcal{L}^b_{\mathcal{S}}(V)$ is continuous with respect to $\sigma(\widetilde{V})|_V$.
Now this $f$ is ($\mathcal{T}|_V$)-continuous since $\mathcal{T}|_V = \mathcal{T}(V)$, therefore Theorem \ref{thm:extension_continuous_to_continuous} implies that $f$ extends to a $\mathcal{T}$-continuous $g \in \mathcal{L}(\widetilde{V})$.
This $g$ is also $\sigma(\widetilde{V})$-continuous by the definition of $\sigma(\widetilde{V})$, therefore $f = g|_V$ is continuous with respect to $\sigma(\widetilde{V})|_V$, as desired.
Hence the lemma holds.
\end{proof}

In what follows, continuity of a map from $\widetilde{V}$ is considered with respect to $\sigma(\widetilde{V})$ instead of $\mathcal{T}$ unless otherwise specified.
Let $\widetilde{V}_0$ denote the intersection of the kernels $\ker(f)$ of all $f \in \mathcal{L}_c(\widetilde{V})$.
Let $\pi$ denote the quotient map $\widetilde{V} \to \widetilde{V}/\widetilde{V}_0$, and let $\widetilde{\mathcal{T}} = \pi(\sigma(\widetilde{V}))$ denote the quotient topology on $\pi(\widetilde{V})$ induced by $\sigma(\widetilde{V})$.
Note that for any $f \in \mathcal{L}_c(\widetilde{V})$, there exists a unique $\overline{f} \in \mathcal{L}_c(\pi(\widetilde{V}))$ such that $f = \overline{f} \circ \pi$, and any element of $\mathcal{L}_c(\pi(\widetilde{V}))$ is obtained in this manner.
Thus by Proposition \ref{prop:weak_top_quotient}, the topology $\widetilde{\mathcal{T}}$ of $\pi(\widetilde{V})$ is a weak topology and coincides with $\sigma(\pi(\widetilde{V}),\mathcal{F})$ where $\mathcal{F} = \{\overline{f} \mid f \in \mathcal{L}_c(\widetilde{V})\}$, therefore $\pi(\widetilde{V})$ is a l.c.t.v.s.\ that is Hausdorff by the definition of $\pi(\widetilde{V})$.
Note that $\pi(V)$ is a linear subspace of $\pi(\widetilde{V})$ and $\pi(\mathcal{S})$ is convex in $\pi(V)$.
Similarly, $\pi(\widetilde{\mathcal{S}})$ is also convex in $\pi(\widetilde{V})$, and $\pi(\widetilde{\mathcal{S}})$ is compact since $\widetilde{\mathcal{S}}$ is compact and $\pi$ is continuous.
On the other hand, since $\sigma(\widetilde{V}) \subset \mathcal{T}$, $V$ is $\mathcal{T}$-dense in $\widetilde{V}$ and $\mathcal{S}$ is ($\mathcal{T}|_{\widetilde{\mathcal{S}}}$)-dense in $\widetilde{\mathcal{S}}$, it follows that $V$ is also $\sigma(\widetilde{V})$-dense in $\widetilde{V}$ and $\mathcal{S}$ is also ($\sigma(\widetilde{V})|_{\widetilde{\mathcal{S}}}$)-dense in $\widetilde{\mathcal{S}}$, therefore $\pi(V)$ is dense in $\pi(\widetilde{V})$ and $\pi(\mathcal{S})$ is dense in $\pi(\widetilde{\mathcal{S}})$ since $\pi$ is continuous.
Moreover, we have the following two properties:
\begin{lemma}
\label{lem:relative_topology_pi_V}
We have $\widetilde{\mathcal{T}}|_{\pi(V)} = \sigma(\pi(V),\mathcal{L}^b_{\pi(\mathcal{S})}(\pi(V)))$.
\end{lemma}
\begin{proof}
Since $\widetilde{\mathcal{T}}|_{\pi(V)}$ is a weak topology by Proposition \ref{prop:weak_top_quotient}, it suffices to show that an $f \in \mathcal{L}(\pi(V))$ is ($\widetilde{\mathcal{T}}|_{\pi(V)}$)-continuous if and only if $f \in \mathcal{L}^b_{\pi(\mathcal{S})}(\pi(V))$.
First, let $f \in \mathcal{L}^b_{\pi(\mathcal{S})}(\pi(V))$.
Then $f \circ \pi|_V \in \mathcal{L}^b_{\mathcal{S}}(V)$, therefore $f \circ \pi|_V \in \mathcal{L}_c(V)$ by the definition of the topology of $V$.
By Lemma \ref{lem:induced_on_S_coincide}, $f \circ \pi|_V$ is also ($\sigma(\widetilde{V})|_V$)-continuous.
Thus Theorem \ref{thm:extension_continuous_to_continuous} implies that $f \circ \pi|_V$ extends to a $g \in \mathcal{L}_c(\widetilde{V})$.
Take the $\overline{g} \in \mathcal{L}_c(\pi(\widetilde{V}))$ corresponding to $g$.
Then we have $\overline{g}(\pi(v)) = g(v) = f(\pi(v))$ for any $v \in V$, therefore $\overline{g}|_{\pi(V)} = f$.
Thus $f$ is ($\widetilde{\mathcal{T}}|_{\pi(V)}$)-continuous.

Secondly, let $f \in \mathcal{L}(\pi(V))$ that is ($\widetilde{\mathcal{T}}|_{\pi(V)}$)-continuous.
Then by Theorem \ref{thm:extension_continuous_to_continuous}, this $f$ extends to a $g \in \mathcal{L}_c(\pi(\widetilde{V}))$.
Now $g \circ \pi \in \mathcal{L}_c(\widetilde{V})$, therefore $B = g \circ \pi(\widetilde{\mathcal{S}})$ is bounded in $\mathbb{R}$ since $\widetilde{\mathcal{S}}$ is compact.
Moreover, we have $f(\pi(s)) = g(\pi(s)) \in B$ for each $s \in \mathcal{S}$, therefore $f(\pi(\mathcal{S})) \subset B$ is also bounded in $\mathbb{R}$.
Thus we have $f \in \mathcal{L}^b_{\pi(\mathcal{S})}(\pi(V))$.
Hence Lemma \ref{lem:relative_topology_pi_V} holds.
\end{proof}
\begin{lemma}
\label{lem:pi_V_bijective}
$\pi|_V$ is a bijection from $V$ to $\pi(V)$.
\end{lemma}
\begin{proof}
Let $v$ and $v'$ be distinct elements of $V$.
Then, since $V$ is Hausdorff by Lemma \ref{lem:Hausdorff_V} and the topology of $V$ is a weak topology, there exists an $f \in \mathcal{L}_c(V)$ such that $f(v) \neq f(v')$.
Now Lemma \ref{lem:induced_on_S_coincide} and Theorem \ref{thm:extension_continuous_to_continuous} imply that this $f$ extends to a $g \in \mathcal{L}_c(\widetilde{V})$, and we have $g(v) \neq g(v')$.
Thus $v - v' \not\in \widetilde{V}_0$ and $\pi(v) \neq \pi(v')$.
Hence the lemma holds.
\end{proof}
By Lemma \ref{lem:relative_topology_pi_V}, Lemma \ref{lem:pi_V_bijective}, and the definition of $\mathcal{T}(V)$, the map $\pi|_V$ is an isomorphism of t.v.s.\ from $V$ to $\pi(V)$.
Moreover, $\pi|_{\mathcal{S}}:\mathcal{S} \to \pi(\mathcal{S})$ is also an isomorphism of convex structures.
The current situation is summarized as follows:
\begin{itemize}
\item $\pi(\widetilde{V})$ is a l.c.\ Hausdorff t.v.s.\ with a weak topology;
\item $\pi(V)$ is a topological vector subspace of $\pi(\widetilde{V})$, with induced topology equal to $\sigma(\pi(V),\mathcal{L}^b_{\pi(\mathcal{S})}(\pi(V)))$, that is dense in $\pi(\widetilde{V})$;
\item $\pi(\mathcal{S}) \simeq \overline{\mathcal{S}_0}$ is a convex subset of $\pi(V)$ that contains the origin of $\pi(V)$ and satisfies $\mathrm{Aff}(\pi(\mathcal{S})) = \pi(V)$, with the relative topology $\sigma(\pi(\mathcal{S}),\mathcal{E}(\pi(\mathcal{S})))$;
\item $\pi(\widetilde{\mathcal{S}})$ is the closure of $\pi(\mathcal{S})$ in $\pi(\widetilde{V})$ that is convex and compact.
\end{itemize}

Note that the above objects $\pi(\mathcal{S})$, $\pi(V)$, $\pi(\widetilde{\mathcal{S}})$, and $\pi(\widetilde{V})$ will be the desired objects in Theorem \ref{thm:state_space_embedded} if the affine hull of $\pi(\widetilde{\mathcal{S}})$ coincides with $\pi(\widetilde{V})$.
However, this is not necessarily guaranteed in general.
Instead, we take a linear subspace $W = \mathrm{Aff}(\pi(\widetilde{\mathcal{S}}))$ of $\pi(\widetilde{V})$ (note that $\pi(\widetilde{\mathcal{S}})$ contains the origin of $\pi(\widetilde{V})$).
Then $W$ is also a l.c.\ Hausdorff t.v.s., and the topology of $W$ is also a weak topology by Proposition \ref{prop:weak_top_quotient}.
This $W$ contains $\pi(V)$ since $\pi(V) = \mathrm{Aff}(\pi(\mathcal{S}))$, and $\pi(V)$ is dense in $W$ since it is dense in $\pi(\widetilde{V})$.
On the other hand, $\pi(\widetilde{\mathcal{S}})$ is also the compact closure of $\pi(\mathcal{S})$ in $W$ since $\pi(\widetilde{\mathcal{S}}) \subset W$.
Moreover, by taking the completion $X$ of the Hausdorff uniform space $\pi(\widetilde{\mathcal{S}})$, the compact subset $\pi(\widetilde{\mathcal{S}})$ of the Hausdorff space $X$ is closed in $X$, therefore $X = \mathrm{cl}_X(\pi(\widetilde{\mathcal{S}})) = \pi(\widetilde{\mathcal{S}})$ and $\pi(\widetilde{\mathcal{S}})$ itself is complete.
Thus the objects $\pi(\mathcal{S})$, $\pi(V)$, $\pi(\widetilde{\mathcal{S}})$, and $W$ play the roles of $\mathcal{S}$, $V$, $\widetilde{\mathcal{S}}$, and $\widetilde{V}$ in Theorem \ref{thm:state_space_embedded}, respectively.
Hence the existence of the objects in Theorem \ref{thm:state_space_embedded} is proven.

\section{Uniqueness of the Objects in Theorem \ref{thm:state_space_embedded}}
\label{sec:appendix_uniqueness}
Finally, we prove the uniqueness of the objects in Theorem \ref{thm:state_space_embedded} (in the sense specified in the statement).
Let $(\mathcal{S},V,\widetilde{\mathcal{S}},\widetilde{V})$ and $(\mathcal{S}',V',\widetilde{\mathcal{S}}',\widetilde{V}')$ be two collections of the objects as in the statement.
First, since $\mathcal{S} \simeq \overline{\mathcal{S}_0} \simeq \mathcal{S}'$, there exists an affine isomorphism $f:\mathcal{S} \to \mathcal{S}'$.
Since $V = \mathrm{Aff}(\mathcal{S})$ and $V' = \mathrm{Aff}(\mathcal{S}')$, this $f$ extends to an affine isomorphism $V \to V'$, denoted also by $f$ (thus $f(\mathcal{S}) = \mathcal{S}'$).
Now note that the topology $\mathcal{T}(V)$ of $V$ is also the weakest topology to make every \emph{affine} functional $g$ on $V$, such that $g(\mathcal{S})$ is bounded in $\mathbb{R}$, a continuous map.
The same also holds for $V'$.
Moreover, for each affine functional $g$ on $V$, $g(\mathcal{S})$ is bounded if and only if $g \circ f^{-1}(\mathcal{S}')$ is bounded.
Thus it follows from the above properties of $\mathcal{T}(V)$ and $\mathcal{T}(V')$ that the affine isomorphism $f:V \to V'$ is also a homeomorphism of topological spaces.

From now, we show that this $f:V \to V'$ extends to the map $\widetilde{V} \to \widetilde{V}'$ specified in Theorem \ref{thm:state_space_embedded}.
For the purpose, take the completions $W$ and $W'$ of $\widetilde{V}$ and of $\widetilde{V}'$, respectively (cf., Appendix \ref{sec:appendix_completion_S_V}).
Then $W$ is also a Hausdorff t.v.s.\ and contains $\widetilde{V}$ (hence $V$) as a dense topological vector subspace.
The same also holds for $W'$ and $\widetilde{V}'$.
Since $W$ and $W'$ are complete, $V$ is dense in $W$, and $V'$ is dense in $W'$, it follows that the above homeomorphism $f:V \to V'$ extends to a homeomorphism $W \to W'$, denoted also by $f$.
Now we have the following:
\begin{lemma}
\label{lem:isom_W_to_W'}
The above map $f:W \to W'$ is also an affine isomorphism.
\end{lemma}
\begin{proof}
It suffices to show that $f$ preserves the convex combination of two elements.
Let $\lambda,\mu \geq 0$ such that $\lambda + \mu = 1$.
Then for each $v,v' \in V$, we have $\lambda f(v) + \mu f(v') = f(\lambda v + \mu v')$ since $f|_V:V \to V'$ is affine.
This implies that the two maps $g_1(v,v') = \lambda f(v) + \mu f(v')$ and $g_2(v,v') = f(\lambda v + \mu v')$ from $V \times V$ to $W'$ coincide with each other.
Since $V \times V$ is dense in $W \times W$ and $W'$ is complete, the continuous map $g_1 = g_2:V \times V \to W'$ has a unique continuous extension $W \times W \to W'$.
On the other hand, both $\overline{g_1}(w,w') = \lambda f(w) + \mu f(w')$ and $\overline{g_2}(w,w') = f(\lambda w + \mu w')$ are continuous maps from $W \times W$ to $W'$ and satisfy that $\overline{g_1}|_{V \times V} = g_1$ and $\overline{g_2}|_{V \times V} = g_2$.
This implies that $\overline{g_1} = \overline{g_2}$, therefore $f(\lambda w + \mu w') = \lambda f(w) + \mu f(w')$ for any $w,w' \in W$.
Hence the lemma holds.
\end{proof}

Since $\widetilde{\mathcal{S}} = \mathrm{cl}_{\widetilde{V}}(\mathcal{S})$ is compact, $\widetilde{\mathcal{S}}$ is also closed in $W$, therefore $\mathrm{cl}_W(\mathcal{S}) = \widetilde{\mathcal{S}}$.
Similarly, we have $\mathrm{cl}_{W'}(\mathcal{S}') = \widetilde{\mathcal{S}'}$.
Since $f:W \to W'$ is a homeomorphism and $f(\mathcal{S}) = \mathcal{S}'$, we have $f(\widetilde{\mathcal{S}}) = \widetilde{\mathcal{S}}'$.
Moreover, since $f:W \to W'$ is an affine isomorphism, $\widetilde{V} = \mathrm{Aff}(\widetilde{\mathcal{S}})$, and $\widetilde{V}' = \mathrm{Aff}(\widetilde{\mathcal{S}}')$, we have $f(\widetilde{V}) = \widetilde{V}'$.
Thus $f|_{\widetilde{V}}:\widetilde{V} \to \widetilde{V}'$ is the desired map specified in Theorem \ref{thm:state_space_embedded}.
Hence the proof of Theorem \ref{thm:state_space_embedded} is concluded.

\end{document}